\documentclass[a4paper,10pt]{article}

\usepackage[utf8]{inputenc}
\usepackage{a4wide}
\usepackage{mathtools,amssymb,amsmath,amsthm}
\usepackage{xcolor}

\usepackage[ocgcolorlinks,colorlinks=true,linkcolor=black,urlcolor=blue,citecolor=green!50!black]{hyperref}
\usepackage{authblk}

\usepackage[capitalize]{cleveref}

\usepackage{csquotes}

\usepackage{navigator}
\embeddedfile{maplemw}{NoyauDoubleCerle.mw}
\embeddedfile{maplepdf}{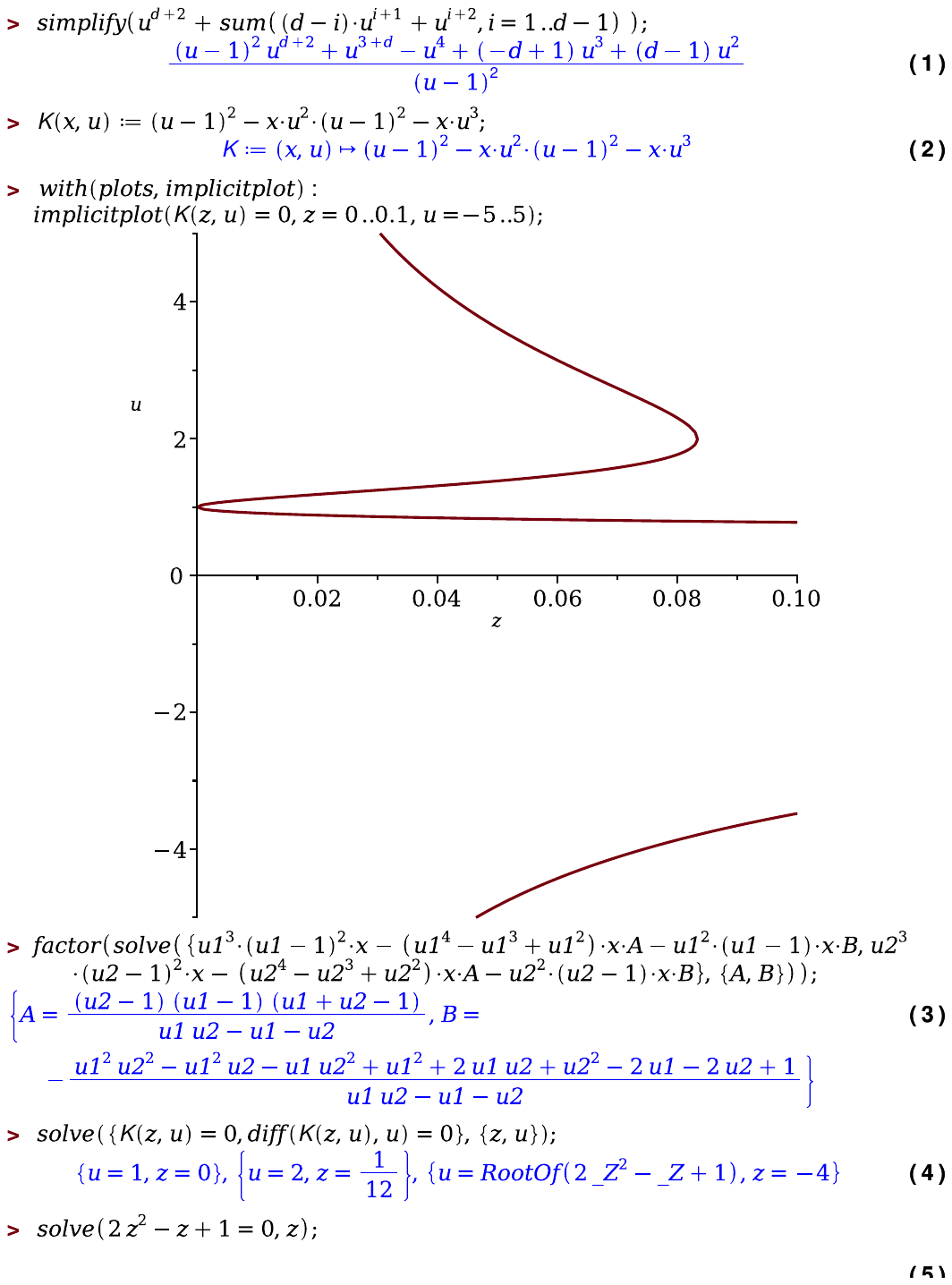}

\newtheorem{definition}{Definition}[section]

\newtheorem{proposition}[definition]{Proposition}
\newtheorem{corollary}[definition]{Corollary}
\newtheorem{lemma}[definition]{Lemma}
\newtheorem{theorem}[definition]{Theorem}
\newcommand{\R}[0]{\mathbb{R}}
\newcommand{\C}[0]{\mathbb{C}}
\newcommand{\RR}[0]{\mathbb{R}}
\newcommand{\eqdef}{\stackrel{\text{def}}=}
\newcommand{\pth}[1]{\left(#1\right)}
\renewcommand{\cal}{\mathcal}
\theoremstyle{remark}
\newtheorem{remark}[definition]{Remark}

\newcommand{\wchi}{\widetilde{\chi}}

\newcommand{\p}{\mathfrak{p}}

\DeclareMathOperator{\DC}{DC}
\DeclareMathOperator{\vect}{Vect}

\title{Large chirotopes with computable numbers of triangulations}
\date{}

\author[1]{Mathilde Bouvel}

\affil[1,3]{Université de Lorraine, CNRS, INRIA, LORIA, F-54000 Nancy, France}
 
\author[2]{Valentin Féray}
 \affil[2]{Université de Lorraine, CNRS, IECL, F-54000 Nancy, France}
 
 \author[3]{Xavier Goaoc}
 \author[4]{Florent Koechlin}
 \affil[4]{CNRS, Université Sorbonne Paris Nord, LIPN, F-93430 Villetaneuse, France}
 
\begin{document}

\maketitle

\begin{abstract}
{Chirotopes are a common combinatorial abstraction of (planar) point sets. In this paper we investigate decomposition methods for chirotopes,} and their application to the problem of counting the number of triangulations supported by a given planar point set.
In particular, we generalize the convex and concave sums operations defined by Rutschmann and Wettstein for a particular family of chirotopes (which they call chains),
and obtain a precise asymptotic estimate for the number of triangulations of the double circle, using a functional equation and the kernel method.
\end{abstract}

\section{Introduction}

Oriented matroids are combinatorial structures that can be used to describe a variety of objects, from acyclic orientations of graphs to the mutual position of points, vectors or hyperplanes in $\R^d$. Here we consider the oriented matroids that arise from point configurations in the plane, given by their {\em chirotope}, that is the function mapping each ordered triple of points to its orientation. Formally, the {\em chirotope of a set} ${\cal P} = \{\p_\ell\}_{\ell
  \in X}$ of points in general position labeled by $X$ is the function
\[ \chi_{\cal P}: \left\{\begin{array}{rcl}
(X)_3 & \to & \{-1,+1\}\\
(x,y,z) & \mapsto & \left\{\begin{array}{rl}
+1& \text{ if $\p_x,\p_y,\p_z$ are in counterclockwise order,}\\
-1 &\text{ if $\p_x,\p_y,\p_z$ are in clockwise order.}\\
\end{array}\right.
\end{array}\right.\]
Here, $(X)_3$ is the set of triples $(x,y,z)$ of distinct elements in $X$. 
This function encodes the {\em labeled order type}~\cite{GP83} of the
point set. We say that chirotopes of point sets are {\em realizable}
to distinguish them from their combinatorial (abstract)
generalization ({see} \cref{s:context}).

In this paper, we investigate a decomposition method for chirotopes, that is, a process to describe large rooted chirotopes in terms of smaller parts. (A {\em rooted} chirotope is a chirotope with a distinguished label.) This method is based on two operations, {\em joins} and {\em meets}, which generalize similar operations on chains defined by Rutschmann and Wettstein~\cite{rutschmann2023chains}. We make the following contributions:
\begin{itemize}
\item We prove that the join and meet operations, which are natural for point sets, can be defined at the abstract level of (rooted) chirotopes. In particular, the join/meet of two realizable chirotopes is again a realizable chirotope (Propositions~\ref{prop:join_realizable} and~\ref{prop:meet_realizable}). 
\item {We associate to each rooted chirotope a bivariate polynomial such that (i) the number of triangulations of the chirotope can be easily recovered from this bivariate polynomial (Lemma~\ref{lem:identifying_true_triangulations}){, and (ii)} these bivariate polynomials can be recursively computed along join and meet 
operations (Proposition~\ref{prop:triangulation_join}).}
\item As an application, we obtain precise asymptotic formulas for the number of triangulations of the double circle $\DC_n$ with $n$ external (and thus $n$ internal) points (\cref{thm:triangulations_double_circle}).
\end{itemize}

\noindent
Throughout the paper, all point sets are finite, planar and in general position, meaning that {\bf no three points are ever aligned} (for simplicity). In other words, all (realizable) chirotopes we consider are {\bf simple}.

\subsection{Context and related work}\label{s:context}

 \paragraph{(Abstract) chirotopes.}
 
The chirotopes of sets of labeled points are called {\em realizable chirotopes} and are special cases of purely combinatorial objects which we now define. Given a finite set $X$, we write $(X)_3$ for the set of triples $(x,y,z)$ of distinct elements in $X$. A {\em chirotope} on a finite set $X$ is a function $\chi: (X)_3 \to \{-1,1\}$ that
satisfies the following properties:\medskip
\begin{description}
\item[(symmetry)] for any distinct $x, y, z \in X$,
  \begin{equation}\label{eq:symmetry_condition}\chi(x,y,z)=\chi(y,z,x)=\chi(z,x,y)=-\chi(z,y,x)=-\chi(y,x,z)=-\chi(x,z,y);\end{equation}
\item[(interiority)] for any distinct $t, x, y, z \in X$,
  \begin{equation} \chi(t,y,z)=\chi(x,t,z)=\chi(x,y,t)=1 \quad \Rightarrow \quad \chi(x,y,z)=1;
  \end{equation}
\item[(transitivity)] for any distinct $s,t, x, y, z \in X$,
  \begin{equation} \chi(t,s,x)=\chi(t,s,y)=\chi(t,s,z) = \chi(x,y,t) = \chi(y,z,t)=1 \quad \Rightarrow \quad \chi(x,z,t) =1.\end{equation}
\end{description}
Actually, such a map is a {\em simple} chirotope since no triple is mapped to $0$, but we drop the word ``simple'' throughout the paper. One can check by elementary geometric arguments that every realizable chirotope is indeed a chirotope. We may use the term {\em abstract chirotope} to emphasize that a chirotope is not necessarily realizable. 
Finally, functions from $(X)_3$ to $\{-1,1\}$ satisfying the symmetry axiom but not necessarily the interiority and transitivity ones will be referred to as {\em sign functions}.

\paragraph{Geometric properties for abstract chirotopes.}

Geometric notions that can be expressed only through orientations of triples of points generalize from planar point sets to abstract chirotopes. Let us give two important example. First, an element $x \in X$ is {\em extreme} in a chirotope $\chi$ on $X$ if there exists $y \in X\setminus\{x\}$ such that $\chi(x,y,z)$ is the same for all $z \in X\setminus \{x,y\}$. With this definition, extreme elements of realizable chirotopes correspond precisely to the vertices of their convex hulls. Second, given a point set $P=\{\p_\ell\}_{\ell \in X}$, two segments $\p_x\p_y$ and $\p_z\p_t$ with distinct endpoints in $P$ cross if and only if $\chi_P(x,y,z) = - \chi_P(x,y,t)$ and $\chi_P(z,t,x) = -\chi_P(z,t,y)$. We can therefore define the crossing of segments for abstract chirotopes: a {\em segment} in a chirotope $\chi$ on $X$ is a pair of elements of $X$, and the segments $xy$ and $zt$ {\em cross in $\chi$} if they satisfy the above condition (with $\chi$ instead of $\chi_P$). 
xe

\paragraph{Numbers of triangulations}

A {\em triangulation} of a chirotope $\chi$ is an inclusion-maximal family of segments such that no two cross in $\chi$. This naturally extends the notion of triangulation of a point set. Let $\mathfrak{T}(\chi)$ denote the set of triangulations of the chirotope $\chi$. When $\chi$ is the chirotope of $n$ points in convex position, $|\mathfrak{T}(\chi)|$ is the $(n-2)$th Catalan number. 
{In general, the problem of maximizing or minimizing the number of triangulations among all possible sets of $n$ points is open (even at an asymptotic level).
Indeed,} for every realizable chirotope $\chi$ of size $n$, it is proven that $|\mathfrak{T}(\chi)|$ is at most $30^n$~\cite{countingT} and is at least $\Omega(2.63^n)$~\cite{min}. 
{By comparison,} the known construction with the largest number of triangulations is the so-called {\em Koch chain} $K_s$, with $n=2^s+1$ points and $|\mathfrak{T}(K_s)| = \Omega(9.08^n)$~\cite{rutschmann2023chains} (more below). 
{As for the lower bound,} the known construction with the smallest number of triangulations is the double circle $\DC_{n}$, with $2n$ points and $|\mathfrak{T}(\DC_{n})| = (12+o(1))^{n}$~\cite{hurtado1997double-circle}.\footnote{\label{footnote1}To be more specific, this result is often mentioned in the literature with reference to~\cite{hurtado1997double-circle}; however, we could not find it precisely in~\cite{hurtado1997double-circle}, neither can we derive it simply from their signed explicit formula. Nevertheless, we shall compute the precise asymptotic of $|\mathfrak{T}(\DC_{n})|$ in \cref{sec:DC}, confirming this statement.}
The double circle is actually conjectured to minimize the number of triangulations among $n$-point realizable chirotopes.

\paragraph{(Koch) chains, convex and concave sums.}

Let us outline some of the ideas behind the construction of the Koch chain (in a language better suited for our purpose). A segment in a chirotope $\chi$ is {\em unavoidable} if it belongs to every triangulation of $\chi$. A {\em chain} {is a chirotope $\chi$ on $X$ such that the elements of $X$ can be arranged into a sequence $x_1,x_2, \ldots, x_n$ where} each pair $\{x_i,x_{i+1}\}$ is an unavoidable segment in $\chi$. Ruschmann and Wettstein~\cite{rutschmann2023chains} defined two operations (shown in \cref{fig:sums}), called {\em convex sum} ($\vee$) and concave sum ($\wedge$), that combine two smaller chains into a larger one. {They also introduce a {\em flip} operation $\bar \cdot$, which simply consists in reversing all orientations.}

\begin{figure}[ht]
\[\includegraphics{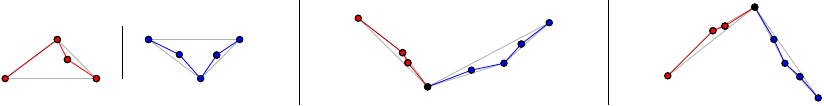}\]
\caption{Two chains (left), their convex sum (center) and their concave sum (right).}
\label{fig:sums}
\end{figure}

\noindent
{An important simplification in considering chains is that every triangulation can be decomposed as a lower part and an upper part, separated by the unavoidable edges. Hence one can study the lower and upper part separately.
Introducing appropriate generating polynomials for the lower and upper parts, it is possible to
compute the number of triangulations by a recursive procedure}~\cite[Lemmas~22 and~24]{rutschmann2023chains}.

We now define the Koch chain, following~\cite{rutschmann2023chains}. Let $E$ denote the chirotope consisting of two points. The Koch chain is defined recursively, by alternating concave and convex sums as follows:

\[ K_2 = E \qquad \text{and} \qquad K_{s+2} = \pth{K_s \wedge K_s} \vee \pth{K_s \wedge K_s}.\]

\noindent
The chirotope $K_{2s}$ is thus a chain with $2^{2s}$ segments, and thus a chirotope on $2^{2s}+1$ elements. As discussed above, for large sizes, this is the known family of chirotopes with the larger number of triangulations.

\paragraph{Relation to modular decomposition of chirotopes.}

This work is in some sense a sequel of~\cite{sidmapapiercombi}, which adapts modular decomposition techniques to chirotopes. Both joins/meets and modular decompositions allow to perform computations on some subclass of (large) chirotopes by breaking them down into smaller chirotopes. This is showcased on the problem of counting triangulations. Also, both techniques open the way to the application of methods from analytic combinatorics. Note however that the classes of decomposable chirotopes are distinct: the double circle with one internal point removed is indecomposable for the modular decomposition but decomposable for joins/meets, and conversely for the chirotope shown in~\cite[Figure~1]{sidmapapiercombi}. Let us stress that unlike for modular decomposition, we do not explore {here} structural properties of joins/meets such as the existence of canonical decompositions or the proportion of indecomposables.

\subsection{Our results}\label{s:results}

In this paper we work with rooted chirotopes. A {\em rooted chirotope} is a pair $(\chi,u)$, where $\chi$ is a chirotope and $u$ is an {\em extreme element} of $\chi$. In the case where $\chi$ is realizable, we say that $(\chi,u)$ is a {\em rooted realizable chirotope}. Given $(\chi,u)$, we denote by $X$ the set of non-root elements of $\chi$, i.e.~$X$ does not contain $u$ and $\chi$ is a chirotope on $X \cup \{u\}$. We also denote by $u^+$ and $u^-$, respectively, the successor and predecessor  of $u$, in counterclockwise order, on the convex hull of $\chi$. Here is a summary of our contributions.

\paragraph{Join and meet operations.}

We first define generalizations of the convex and concave sums, which we call the join and meet operations; they act not only on chains, but on arbitrary rooted chirotopes.

\begin{figure}[ht]
\[\includegraphics{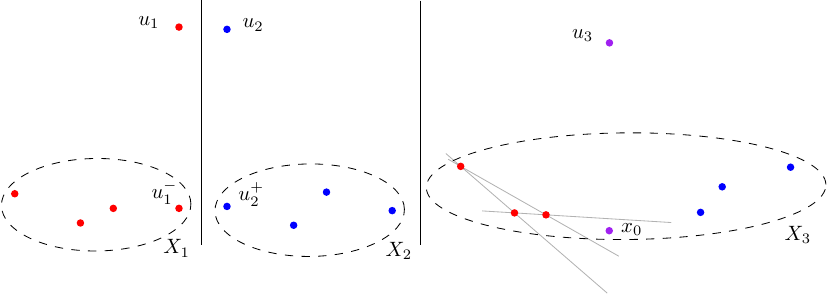}\]
\caption{Example of join of two (realizable) chirotopes.}
\label{fig:join}
\end{figure}

To give a geometric intuition of the join operation, consider two point sets $X_1\cup\{u_1\}$ and $X_2\cup\{u_2\}$ with $u_i$ extreme in $X_i$ for $i=1,2$. Their join is, informally, a merge $X_3\cup\{u_3\}$ of $X_1\cup\{u_1\}$ and $X_2\cup\{u_2\}$ where, as illustrated in Figure~\ref{fig:join}:
\begin{itemize}
	\item $u_1$ and $u_2$ are identified as $u_3$,
	\item $u_1^{-}$ and $u_2^{+}$ are identified as $x_0$. 
	\item no line through two points of $X_1$ separates $u_3$ from any point of $X_2\setminus\{u_2^+\}$, and, symetrically, no line through two points of $X_2$ separates $u_3$ from any point of $X_1\setminus\{u_1^-\}$.
\end{itemize}

\begin{figure}[ht]
\[\includegraphics{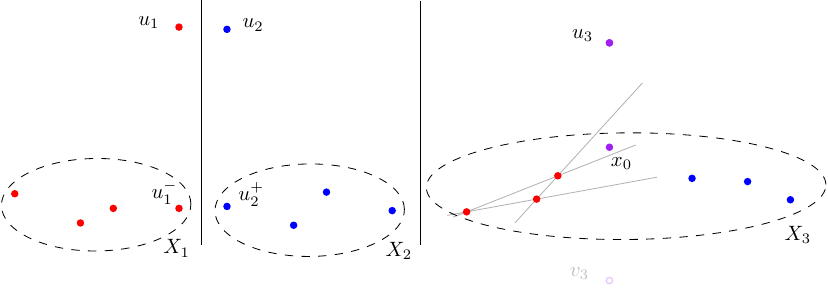}\]
\caption{Example of meet of two (realizable) chirotopes.}
\label{fig:meet}
\end{figure}

Similarly, as illustrated in \cref{fig:meet}, in the meet $X_3\cup\{u_3\}$ of $X_1\cup\{u_1\}$ and $X_2\cup\{u_2\}$ we have that:

\begin{itemize}
	\item $u_1$ and $u_2$ are identified as $u_3$,
	\item $u_1^{-}$ and $u_2^{+}$ are identified as $x_0$. 
	\item every line through two points of $X_1$ separates $u_3$ from every point of $X_2\setminus\{u_2^+\}$, and, symetrically, every line through two points of $X_2$ separates $u_3$ from every point of $X_1\setminus\{u_1^-\}$.
\end{itemize}

\noindent
In Section~\ref{sec:Join_Meet}, we give formal definitions and prove that the join or meet of (realizable) rooted chirotopes is a (realizable) rooted chirotope (Propositions~\ref{prop:join_realizable} and~\ref{prop:meet_realizable}). Comparing Figures~\ref{fig:join} and~\ref{fig:meet} with Figure~\ref{fig:sums}, it should not come as a surprise that the join of two chains is their convex sum, and the meet of two chains is their concave sum; see \cref{rmk:join=convexSum} for detail.
A difficulty in proving the realizability of joins and meets for general chirotopes compared to the case of chains is that we need to construct appropriate realizations of our chirotopes {(where, informally, $u_1$, resp.~$u_2$, lies ``on top'' of all points in $X_1$, resp.~$X_2$)} before combining them into a single point set. In the case of chains, these appropriate realizations are given by definition.

\paragraph{Counting triangulations of joins and meets} 
We can associate to a rooted chirotope $(\chi,u)$ a polynomial, called its triangulation polynomial and denoted $Q_{(\chi,u)}$, whose coefficients are the numbers of triangulations of $(\chi,u)$ counted by degree of the root $u$. Accordingly, we typically also use $u$ to denote the variable of this polynomial $Q_{(\chi,u)}$. The polynomial $Q_{(\chi,u)}(u)$ encodes many informations about the triangulations of $(\chi,u)$; in particular, its evaluation at $u=1$ gives the total number of triangulations of $\chi$. We explain  in \cref{sec:triangulations_join} how to compute the triangulation polynomial of the join of two rooted chirotopes from their respective triangulation polynomials. More precisely, we prove that it is of the form
\[
 Q_{{(\chi_1,u_1)} \vee {(\chi_2,u_2)}}(u) = \sum_{d_1 \geq 2} \sum_{d_2 \geq 2} [u^{d_1}] Q_{{(\chi_1,u_1)}}(u) \cdot [u^{d_2}] Q_{{(\chi_2,u_2)}}(u) \cdot N_{d_1,d_2}(u),
\]
where $N_{d_1,d_2}(u)$
 is a polynomial in $u$  depending only on $d_1$ and $d_2$ (see \cref{corol:Qjoin} for an exact formula).  
In particular, this makes it possible to compute the {number of} triangulations of rooted chirotopes built from successive iterated joins (or similarly iterated meets).

In order to compute the {number of} triangulations of chirotopes built using both meet and join operations together, it is necessary to generalize the triangulation polynomial. For this, we intuitively add an extra fantom element $v$ to the rooted chirotope, that acts as the opposite of $u$ in terms of orientations (see for instance the point $v_3$ in \cref{fig:meet}). We can then define a bivariate triangulation polynomial, following the degrees of both $u$ and $v$ in triangulations avoiding the edge $uv$ (called weak triangulations). This allows to compute the triangulation polynomials of any chirotope built from join and meet operations.

{In the special case of chains, our bivariate triangulation polynomial essentially reduces to the product
of the lower triangulation polynomial in one variable and the upper triangulation polynomial in the other variable.
Our recursive formula to compute the bivariate polynomial through join and meet then corresponds
to the recursive formulas of Ruschmann and Wettstein to compute the lower/upper triangulation polynomials through convex and concave sums. See Remark~\ref{rmk:WeakTforChains=Upper*Lower} for details.
Note however, that for general rooted chirotopes, triangulations cannot be split in upper and lower
parts and thus, we cannot define lower/upper triangulation polynomials. 
This explains why we need in general a bivariate polynomial, instead of a pair of univariate polynomials.}

\paragraph{The number of triangulations of the double circle}

Recall that $\DC_{k}$ denotes the double circle with $2k$ points. As mentioned earlier, it is conjectured that $\DC_{k}$ minimizes the number of triangulations among all  realizable chirotopes with $2k$ points. Although we do not prove this conjecture, we compute in Section~\ref{sec:DC} a precise asymptotics for the number $|\mathfrak{T}(\DC_{k})|$ of triangulations of $\DC_{k}$, namely (see \cref{thm:triangulations_double_circle}): 
\[ |\mathfrak{T}(\DC_k)|  \sim_{k \to \infty}   
\frac{54} {7 \sqrt{\pi}} \, \frac{\sqrt{21}(5-\sqrt{21})}{(7-\sqrt{21})^2}\, 12^{k-2} \, k ^{-3/2}.\]
This is much finer than the estimate $|\mathfrak{T}(\DC_{k})| = (12+o(1))^{k}$ attributed to~\cite{hurtado1997double-circle} (see footnote~\ref{footnote1}).

To obtain this result, we first relate the number of triangulations of $\DC_{k}$ to the number of triangulations of some variants of the double circle (see \cref{lem:nbTriDoubleCircle}), 
{one of which can be recursively constructed via join operations}. 
{Building on this,} 
{we obtain that} $|\mathfrak{T}(\DC_k)|$ can be expressed thanks to a particular polynomial $Q_{k-1}(u)$, and these polynomials satisfy a recurrence (on $k$) derived from the results of Section~\ref{sec:triangulations_join} (see \cref{lem:induction_Qk}). Introducing the bivariate generating function $F(z,u)=\sum_{k \ge 1} Q_k(u) z^k$, this recurrence rewrites as a functional equation of $F(z,u)$ (see \cref{lem:functional_equation_F}). This functional equation is then amenable (with some effort) to the techniques of analytic combinatorics, and in particular the kernel method. From this study, we derive an expansion of $F(z,1)$ near its radius of convergence $1/12$, which yields (through the transfer theorem of analytic combinatorics) an asymptotic estimates for $Q_{k-1}(1)$. In addition, we have a similar statement for $\tfrac{\partial F(z,u)}{\partial u}|_{u=1}$, and combining these two gives access to our asymptotic estimate for $|\mathfrak{T}(\DC_k)|$. 

\paragraph{Numerical experiments}

Counting the triangulations of a planar point set is a classical and challenging problem (see the discussion in~\cite{marx_et_al:LIPIcs.SoCG.2016.52}). Our recursive formula allows to count efficiently and exactly the triangulations of chirotopes built from joins and meets of small chirotopes. We implemented it and used it to try finding chirotopes with more triangulations than the Koch chains. We report in Section~\ref{sec:numeric} on these (unsuccessful) attempts.

\paragraph{Important note} During the preparation of the present paper, we learned about the work of Bui~\cite{bui2025counting} which has a significant intersection with ours (the currently available version of Bui's work unfortunately does not contain proofs of the stated result). In particular, Bui extends Ruschmann--Wettstein's convex and concave sums to a family of point sets called {\em near-edges} and argues that any point set is equivalent to a near-edge. Though the presentation is different, this is essentially equivalent to the construction that we make in \cref{sec:Join_Meet}. Bui also associates bivariate polynomials to chirotopes and provides recursive formulas to compute these polynomials when performing joins and meets; again, the presentation is different but may be roughly equivalent to our results in Section~\ref{sec:triangulations_join}. On the other hand, precise asymptotics for the number of  triangulations of the double circle is not considered in Bui's work. Finally, Bui also performed numerical computations to look for chirotopes with more triangulations than the Koch chain. Here, we used different approaches: while we compared the number of triangulations of chirotopes with the same number of points (namely, 257), Bui uses a conjectural formula for their growth constants. The two approaches complement each other and both suggest that we cannot construct chirotopes with more triangulations than the Koch chain 
with the generalized join and meet operations.

\section{Joins and meets of rooted chirotopes}
\label{sec:Join_Meet}

In this section we formalize the notions of joins and meets  that were described at an intuitive level for realizable chirotopes in Section~\ref{s:results}. More precisely, we define join and meets for rooted chirotopes (not only realizable ones), and we prove that these constructions are well-founded.

\subsection{Join of rooted chirotopes}

\begin{definition}
The join of two rooted chirotopes $(\chi_1,u_1)$ and $(\chi_2,u_2)$ is the pair $(\chi_3, u_3)$, denoted  
$(\chi_3, u_3) := (\chi_1, u_1)\vee (\chi_2, u_2)$, where $(i)$ $\chi_3$ is a sign function with ground set $X_3\cup\{u_3\}$, $(ii)$ $X_3$ is the disjoint union of $X_1$ and $X_2$ in which 
 we identify $u_1^-$ and $u_2^+$, and $(iii)$ $\chi_3$ satisfies 
 \begin{eqnarray}
 	\forall x,y,z \in X_3,\ \chi_3(x,y,z) &=&   \label{eq:def_chi3}
  \begin{cases}
   \chi_1(x,y,z) \text{ \quad if } x,y,z \in X_1, \\
   \chi_2(x,y,z) \text{ \quad if } x,y,z \in X_2, \\
   \chi_1(x,y,u_1) \text{ \quad if } x,y \in X_1 \text{ and } z \in X_2\setminus\{u_2^+\}, \\
   \chi_2(u_2,y,z) \text{ \quad if } x \in X_1\setminus\{u_1^-\} \text{ and } y,z \in X_2;
   \end{cases}\\
   \forall x,y\in X_3,\ \chi_3(x,y,u_3) &=& \label{eq:def_chi3_u3}
  \begin{cases}
   \chi_1(x,y,u_1) \text{ \quad if } x,y \in X_1, \\
   \chi_2(x,y,u_2) \text{ \quad if } x,y \in X_2, \\
    1 \text{ \quad if } x \in X_1\setminus\{{u_1}^-\} \text{ and } y \in X_2\setminus \{{u_2}^+\}.
   \end{cases}
 \end{eqnarray}
\end{definition}

Together with the fact that $\chi_3$ is a sign function, the above equations define
$\chi_3$ for all possible triples of distinct elements in $X_3\cup\{u_3\}$ (the three arguments are either all in $X_1$ or all in $X_2$, or two are in $X_1$ and one in $X_2$ or conversely, or one argument is $u_3$ and then we consider the possible repartition of the two other arguments in $X_1$ and $X_2$).
Let us denote by $x_0$ the element $u_2^+=u_1^-$. When $x_0$ is one of the arguments, several conditions may be fulfilled (up to permuting the arguments) but the corresponding formulas are consistent.
For example,
if $x \in X_1\setminus\{{u_1}^-\}$, $y=x_0$ and $z \in X_2\setminus \{{u_2}^+\}$, then we can use $x_0=u_1^-$
and the third case of \eqref{eq:def_chi3} to obtain
$\chi_3(x,x_0,z) =\chi_1(x,u_1^-,u_1) =+1$,
or we can use $x_0=u_2^+ \in X_2$ 
and the fourth case of \eqref{eq:def_chi3} to obtain
$\chi_3(x,x_0,z)= \chi_1(u_2,u_2^+, z) =+1$. Altogether the join operation is well-defined.

\begin{proposition}
\label{prop:join_realizable}
The join  of two rooted realizable chirotopes 
is a rooted realizable chirotope.
\end{proposition}
We defer the proof of \cref{prop:join_realizable} to \cref{sec:proofjoinrealizable}.

\begin{corollary}\label{cor:join}
	The join of two rooted chirotopes 
is a rooted chirotope.
\end{corollary}
\noindent One way to prove the statement is to deduce the properties for $\chi_3$ from the properties for $\chi_1, \chi_2$ via substitutions. This is straightforward but tedious (see for instance \cite[Prop. 3.1]{sidmapapiercombi} for a proof in this spirit
that a different operation on sign functions preserve chirotopes).
We provide here a less pedestrian alternative.
\begin{proof}
Let $(\chi_1,u_1)$ and $(\chi_2,u_2)$ be rooted chirotopes with ground sets $X_1 \cup \{u_1\}$ and $X_2 \cup \{u_2\}$ respectively, and let $(\chi_3,u_3)=(\chi_1,u_1)\vee (\chi_2,u_2)$ be their join. We have defined $(\chi_3,u_3)$ as a rooted sign function,
and want to prove that it is a rooted chirotope.
 
Given a sign function $\chi$ with ground set $X$, and $Y\subseteq X$, we denote by $\chi_{|Y}$ the restriction of the function $\chi$ to the triplets of $Y$. 
Remark that the join of two restrictions is the restriction of the join in the following sense: 
\[\forall Y_1\subseteq X_1,\forall Y_2\subseteq X_2 \qquad {\chi_1}_{|Y_1\cup\{u_1,u_1^-\} }\vee {\chi_2}_{|Y_2\cup\{u_2,u_2^+\}} = {\chi_3}_{|Y_1\cup Y_2\cup\{u_3,u_1^-,u_2^+\}}\ .\]
Now, a classical fact about chirotopes is that every chirotope on up to 8 points is realizable (this is essentially Theorem~1 from~\cite{goodman1980proof}). Hence for all $Y_1\subseteq X_1$ of size at most 6, and all $Y_2\subseteq X_2$ of size at most 6, ${\chi_1}_{|Y_1\cup\{u_1,u_1^-\} }$ and $ {\chi_2}_{|Y_2\cup\{u_2,u_2^+\}}$ are realizable chirotopes. \cref{prop:join_realizable} therefore ensures that ${\chi_3}_{|Y_1\cup Y_2\cup\{u_3,u_1^-,u_2^+\}}$ is a (realizable) chirotope. 
As chirotopes can be characterized by Knuth's axioms, each of which involves at most 5 elements, a sign function $\xi$ is a chirotope if and only if for every  five-element subset $Y$ of the ground set of $\xi$, the function ${\xi}_{|Y}$ is a chirotope. 
But we have proved that the latter holds for $\chi_3$, so $\chi_3$ is a chirotope.

It remains to prove that $u$ is an extreme element in $\chi_3$.
Again, it suffices to prove that $u$ is extreme in every restriction of $
\chi_3$ to a four-element subset that contains $u$. (We use here a standard extension of Carathéodory's theorem, see {\em e.g.}~\cite[Lemma 2.1]{sidmapapiercombi}.) 
But this follows from \cref{prop:join_realizable} (using again that such restrictions are necessarily realizable),
concluding the proof of the lemma. 
\end{proof}

\subsubsection{Proof of \cref{prop:join_realizable}}\label{sec:proofjoinrealizable}

We prove that the join of two rooted realizable chirotopes $(\chi_1,u_1)$ and $(\chi_2,u_2)$ is a rooted realizable chirotope by taking the union of suitable realizations of $\chi_1$ and $\chi_2$. We build these ``suitable'' realizations in three steps. 
The first step puts the root point and its neighbors on the convex hull in particular positions.  

\begin{lemma}\label{lem:extreme_in_unbounded_cell}
Every realizable rooted chirotope  $(\chi,u)$ with ground set $X$ has a realization $\mathcal P=\{\mathfrak p_x\}_{x \in X}$ such that $\mathfrak p_u = (0,1)$, $\mathfrak p_{u^+}=(0,0)$, $\mathfrak p_{u^-}=(1,0)$, and the half-line $\mathfrak
p_u+\RR_{\ge 0}(0,1)$ does not cross any line of the 
set $\big\{(\mathfrak p_x\mathfrak p_y) \colon x,y \in X \setminus \{u\},
x \neq y \big\}$.
\end{lemma}
\begin{proof}
	This proof relies on the notion of projective transform, which we now recall; see {\em e.g.} the book of Samuel~\cite{samuel1988projective} or Richter-Gebert~\cite{richter2011perspectives} for a primer on projective geometry.
    For a point $\mathfrak p \in \R^2$, let $\hat{\mathfrak p}$ be the vector $\hat{\mathfrak p} = (p_x,p_y,1)$, and conversely, given $u \in \R^3$ with $u_z \neq 0$ let $\bar u \eqdef (u_x/u_z, u_y/u_z)$. A projective transform is a map $f:\mathfrak p \mapsto \overline{A\hat{\mathfrak p}}$ with $A \in \R^{3 \times 3}$ a nonsingular matrix. Such a transform is undefined on a line of $\R^2$ (those points $\mathfrak p$ such that $\pth{A\hat{\mathfrak p}}_z = 0$); this line is said to be ``sent to infinity''. Other lines of $\mathbb R^2$ are sent to lines in $\mathbb R^2$.\smallskip

    Let us fix some arbitrary realisation $\cal Q$ of $\chi$, and define two auxiliary points. The first point, which we denote by $\mathfrak q_{\infty}$, is chosen on the line $(\mathfrak q_u,\mathfrak q_{u^+})$, outside of the segment $[\mathfrak q_u,\mathfrak q_{u^+}]$ and close enough to $\mathfrak q_u$ so that the segment $[\mathfrak q_u,\mathfrak q_\infty]$ does not cross any line of the set $\big\{(\mathfrak q_x\mathfrak q_y) \colon x,y \in X \setminus \{u\}\big\}$. (In other words, we fix $\mathfrak q_{\infty} = \mathfrak q_u -\lambda (\mathfrak q_{u^+}-\mathfrak q_u)$ with $\lambda>0$ small enough.)  Since  $[\mathfrak q_u,\mathfrak q_{u^+}]$ is an edge of the convex hull of $\cal Q$,
 the point $\mathfrak q_{\infty}$ is outside this convex hull, and we can find a line going through $\mathfrak q_{\infty}$ which does not separate the point set $\cal Q$. We pick our second point $\mathfrak q_{\infty}'$ on that line, distinct from $\mathfrak q_{\infty}$.

 Now, let $A\in\RR^{3\times 3}$ be a nonzero matrix that satisfies 
 \[\begin{array}{ll}
	\text{(i) }  {A\widehat{\mathfrak q_u}}\in \vect((0,1,1)), \qquad&
	\text{(ii) }{A\widehat{\mathfrak q_{u^+}}}\in \vect((0,0,1)), \\
\text{(iii) } {A\widehat{\mathfrak q_{u^-}}}\in \vect((1,0,1)), \qquad&
	\text{(iv) } \{A\widehat{\mathfrak q_{\infty}},A\widehat{\mathfrak q'_{\infty}}\}\subset \{z=0\}.
\end{array}\]
 The existence of $A$ follows from the observation that for any given $\mathfrak q_u, \mathfrak q_{u^+}, \mathfrak q_{u^-}, \mathfrak q_{\infty}, \mathfrak q'_{\infty}$, conditions~(i-iv), when translated into equations for the coefficients of $A$, define a linear system of 8 equations and 9 variables. Hence there exists a nonzero solution.

 We claim (and prove below) that $A$ is actually nonsingular, so that $f:\mathfrak p \mapsto \overline{A\hat{\mathfrak p}}$ is a projective transform. We set $\mathcal P=f(\mathcal Q)$ and note that $\mathfrak p_u = (1,0)$ by (i), $\mathfrak p_{u^+}=(0,0)$ by (ii) and $\mathfrak p_{u^-}=(1,0)$  by (iii). Condition (iv) ensures that the line  $(\mathfrak q_{\infty} \mathfrak q'_{\infty})$ is the line sent to infinity by $f$.
 Moreover, since this line does not split $\cal Q$, the point sets $\cal P$ and $\cal Q$ have the same chirotope, i.e.~$\cal P$ is a realization of $\chi$. We now consider the
 image by $f$ of $[\mathfrak q_{u}, \mathfrak q_\infty]$;
 this is is a half-line $\rho$ originating from $\mathfrak p_u$. 
Recalling that projective transforms preserve alignments and incidences,
since $\mathfrak q_u,\mathfrak q_{u^+},\mathfrak q_{\infty}$ are aligned, $\rho$ is contained on the line $(\mathfrak p_u,\mathfrak p_{u^+})$. But $[\mathfrak q_{u}, \mathfrak q_\infty]$ does not contain $\mathfrak q_{u^+}$, so $\rho$ does not contain $\mathfrak p_{u^+}$. Hence $\rho$ is the half-line $\mathfrak p_u+\RR_{\ge 0}(0,1)$.
In addition, the segment $[\mathfrak q_{u}, \mathfrak q_\infty]$ avoids every line in  $\big\{(\mathfrak q_x\mathfrak q_y) \colon x,y \in X \setminus \{u\},
x \neq y \big\}$, so that $\rho$ avoids every line in $\big\{(\mathfrak p_x\mathfrak p_y) \colon x,y \in X \setminus \{u\},
x \neq y \big\}$. \smallskip

It remains to prove the claim that $A$ is nonsingular. It suffice to argue that none of $A\widehat{\mathfrak q_u}$, $A\widehat{\mathfrak q_{u^+}}$ and $A\widehat{\mathfrak q_{u^-}}$ is equal to $(0,0,0)$, as these three images are then linearly independent.  Assume for the sake of contradiction that one of them, say $A\widehat{\mathfrak q_u}$ is $(0,0,0)$.
Then the matrix $A$ sends all three vectors $\widehat{\mathfrak q_u}$, $\widehat{\mathfrak q_{\infty}}$ and $\widehat{\mathfrak q'_{\infty}}$ 
to the hyperplane $\{z=0\}$. Since the points $ \mathfrak q_u$, $\mathfrak q_{\infty}$
and $\mathfrak q'_{\infty}$ are not colinear (recall that the line $(\mathfrak q_{\infty}\mathfrak q'_{\infty})$ does not split the point set $\cal Q$, so in particular does not contain any point of $\cal Q$), the vectors $\widehat{\mathfrak q_u}$, $\widehat{\mathfrak q_{\infty}}$ and $\widehat{\mathfrak q'_{\infty}}$ are linearly independent. We conclude that the range of $A$ is entirely contained in the hyperplane $\{z=0\}$. But then conditions (ii') and (iii') then imply that $A\widehat{\mathfrak q_{u^+}}$ and $A\widehat{\mathfrak q_{u^-}}$ must be equal to $(0,0,0)$. We have found three linearly independent vectors  $\widehat{\mathfrak q_u}$,
$\widehat{\mathfrak q_{u^+}}$ and $\widehat{\mathfrak q_{u^-}}$ which are mapped to $(0,0,0)$ by $A$, implying that $A$ is the zero matrix and reaching a contraction.
We conclude that none of of $A\widehat{\mathfrak q_u}$, $A\widehat{\mathfrak q_{u^+}}$ and $A\widehat{\mathfrak q_{u^-}}$ is equal to $(0,0,0)$ so that $A$ is nonsingular. The proposition is proved.
\end{proof}

The second step consists in ``flattening'' the point set so that all points but the root
en up close to the horizontal axis. More precisely, we introduce a refinement of the notion of $\delta$-squeezing from \cite[Definition 6.4]{sidmapapiercombi} to formalize this flattening step.

\begin{definition}\label{def:deltaepsilonsqueezing}
  A \emph{right-$(\delta,\varepsilon)$-squeezing} of a rooted chirotope $(\chi,u)$ is a realization $\cal P$ of $\chi$
  such that $\mathfrak p_u = (0,1)$, $\mathfrak p_{u_+}=(0,0)$, $\mathfrak p_{u^-}=(1,0)$, ${\cal P}\setminus\{\mathfrak p_u,\mathfrak p_{u^+}\}\subseteq[\varepsilon,1]\times [-\delta,\delta]$, and every line going through two points of ${\cal P} \setminus \{\mathfrak p_u\}$ has slope in $[-\delta,\delta]$.
\end{definition}

\noindent
We now prove that such $(\delta,\varepsilon)$-squeezings always exist.

\begin{lemma}\label{lem:deltaepsilonsqueezing}
  For every rooted realizable chirotope $(\chi,u)$, there exists $\varepsilon>0$ such that for every $\delta>0$, there exists a  right-$(\delta,\varepsilon)$-squeezing of $(\chi,u)$.
 \end{lemma}
\begin{proof}[Proof of \cref{lem:deltaepsilonsqueezing}]
  Let $\cal Q$ be a realization of $\chi$ as given by \cref{lem:extreme_in_unbounded_cell}. 
  Since $\mathfrak q_{u}\mathfrak q_{u^+}$ is an edge of the convex hull of $\cal Q$,
  all other points in $\cal Q$ have a positive $x$-coordinate.
  Let $\varepsilon >0$ be the smallest $x$-coordinate of a point from $\cal Q\setminus \{\mathfrak q_{u},\mathfrak q_{u^+}\}$. 
  Observe that no point $\mathfrak q$ of $\cal Q$ may have an $x$-coordinate strictly larger than $1$: otherwise, either $\mathfrak q_u \mathfrak q_{u^-}$ would not be an edge of the convex hull (if $\mathfrak q$ is above the horizontal axis), or $ \mathfrak q_{u^-}$ would not be extreme (if $\mathfrak q$ is not too far below the horizontal axis) or the line $(\mathfrak q, \mathfrak q_{u^-})$ would intersect the half-line $\mathfrak q_{u} + \R_{\ge 0} (0,1)$,
  contradicting \cref{lem:extreme_in_unbounded_cell} (if $\mathfrak q$ is too much below the horizontal axis).
 It follows that every point of $\mathcal Q\setminus \{\mathfrak q_u, \mathfrak q_{u^+} \}$ has its $x$-coordinate in $[\varepsilon,1]$.

  Let $t \in \R_{\ge 0}$ and let us define $\mathfrak q_t \eqdef \mathfrak q_u + t (0,1)$ and  $f_{t}: (x,y) \mapsto (x, \frac{1}{1+t} y)$. 
  The map $f_{t}$ preserves orientations. Since $\mathfrak q_t$ and $\mathfrak q_u$ are in the same cell of the arrangement of the set of lines $\big\{(\mathfrak q_x\mathfrak q_y) \colon x,y \in X \setminus \{u\},x \neq y \big\}$, the point set
  $\{\mathfrak q_t\} \cup \cal Q \setminus \{\mathfrak q_u\}$ has the same chirotope as  $\cal Q$, that is $\chi$.
 
  Now consider the point set $ {\cal P}(t) \eqdef \{f_{t}(\mathfrak q) \colon \mathfrak q \in  \{\mathfrak q_t\} \cup \cal Q \setminus \{\mathfrak q_u\}\}$. Since $f_{t}$ preserves orientations, ${\cal P}(t)$ is also a geometric realization of $\chi$. It contains in particular the three following point: $\mathfrak p_u(t) =f_t(\mathfrak q_t)= (0,1)$, $\mathfrak p_{u^+}(t)=f_t(\mathfrak q_{u^+})=(0,0)$, and $\mathfrak p_{u^-}(t)=f_t(\mathfrak q_{u^-})=(1,0)$. Also, the other points, i.e.~the elements of ${\cal P}(t)\setminus\{\mathfrak p_u(t),\mathfrak p_{u^+}(t),\mathfrak p_{u^-}(t)\}$, are all in $[\varepsilon,1]\times \RR$. As $t\to\infty$, the $y$-coordinates of each point in ${\cal P}(t)\setminus\{\mathfrak p_u(t)\}$ goes to $0$; similarly, as $t\to\infty$, the slope of every line going through two points of ${\cal P}(t) \setminus \{\mathfrak p_u(t)\}$ goes to $0$. Henceforth, for every $\delta>0$, for $t$ large enough, ${\cal P}(t)$ is a $(\delta,\varepsilon)$-squeezing of $(\chi,u)$.
\end{proof}

At this point, we would like to take the union of a right-$(\delta,\varepsilon_2)$-squeezing of $(\chi_2,u_2)$ and the mirror image (with respect to the line $x=0$) of a right-$(\delta,\varepsilon_1)$-squeezing of $-\chi_1$. This does not quite prove \cref{prop:join_realizable} yet: we still need, for instance, to ensure that the lines spanned by pairs of points of $\chi_1$ other than $\{u_1,u_1^-\}$ leave all points of $\chi_2$ other than $u_2^+$ on the same side. It turns out that taking $\delta$ small enough does not suffice to ensure this, but combining it with small rotation is enough; this is the third and final step of the construction.

\begin{proof}[Proof of \cref{prop:join_realizable}]
  Let $(\chi_1,u_1)$ and $(\chi_2,u_2)$ be two rooted realizable chirotopes. We want to prove that their join is a rooted realizable chirotope.

  \medskip

  By \cref{lem:deltaepsilonsqueezing}, there exist $\varepsilon_2>0$ and $\varepsilon_1>0$ such that for every $\delta>0$, there exists a right-$(\delta,\varepsilon_2)$-squeezing ${\cal P}^2$ of $(\chi_2,u_2)$ and a right-$(\delta,\varepsilon_1)$-squeezing ${\cal P}^{-1}$ of $(-\chi_1,u_1)$. We let ${\cal P}^1$ denote the mirror image of ${\cal P}^{-1}$ with respect to the line $x=0$. Note that both ${\cal P}^1$ and  ${\cal P}^{2}$ depend on $\delta$, as all constructions from here. We let $\varepsilon = \min(\varepsilon_1,\varepsilon_2)$ and we note that ${\cal P}^1$ is a realization of $\chi_1$. 

  \medskip
  
  We let $B^2$ denote the bounding box $[\varepsilon,1] \times [-\delta,\delta]$ of ${\cal P}^2\setminus\{\p_{u_2}^2,\p_{{u_2}^+}^2\}$ and $L^2$ denote the set of lines with slope in $[-\delta,\delta]$ that intersect $B^2$. Note that every line spanned by two points of ${\cal P}^2 \setminus \{\mathfrak p_{u_2}^2\}$ is in~$L^2$. Similarly, we let $B^1 = [-1,-\varepsilon] \times [-\delta,\delta]$ and let $L^1$ denote the set of lines with slope in $[-\delta,\delta]$ that intersect $B^1$. Note that for $\delta$ small enough, $(0,1)$ is above every line in $L^1 \cup L^2$.
 
\begin{figure}[ht]
\begin{center}
 \includegraphics[width=.98\linewidth,keepaspectratio]{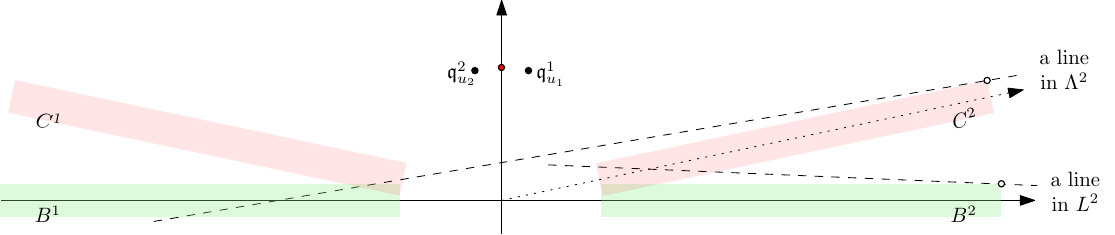}\vspace{-2mm}
\end{center}
  \caption{The bounding boxes.}
\end{figure}

  Let us, momentarily, fix $\delta = \delta_0$ small enough such that $(0,1)$ is above every line in $L^1 \cup L^2$.
We let $\alpha>0$ be some small angle to be defined shortly. We apply the rotation with center $(0,0)$ and angle $\alpha$ to ${\cal P}^2$, $B^2$ and $L^2$, and the rotation with center $(0,0)$ and angle $-\alpha$ to ${\cal P}^1$, $B^1$ and $L^1$. We denote by ${\cal Q}^i$, $C^i$ and $\Lambda^i$ the images of ${\cal P}^i$, $B^i$ and $L^i$, respectively. For $\alpha$ small enough, $(0,1)$ remains above every line of $\Lambda^1 \cup \Lambda^2$. Up to reducing $\alpha$, we can therefore ensure that $\mathfrak q_{u_1}^1$ and $\mathfrak q_{u_2}^2$ are also above every line of $\Lambda^1 \cup \Lambda^2$. We let ${\cal R}^i$ denote ${\cal Q}^i$ in which $\mathfrak q_{u_i}^i$ has been replaced by $\mathfrak p_{u_i}^i=(0,1)$, and note that ${\cal R}^i$ realizes $\chi_i$.

  \medskip

  We now fix $\alpha$ and let $\delta$ vary again. An important observation is that for all $\delta<\delta_0$ the property that ${\cal R}^i$ realizes $\chi_i$ remains. Finally, observe that as  $\delta$ goes to $0$, for every line $\ell \in \Lambda^1$, the slope of $\ell$ goes to $\alpha$ and the distance from $\ell$ to the origin goes to $0$. Since $\varepsilon>0$ and $C^2$ is compact, for $\delta>0$ small enough no line in $\Lambda^1$ intersect $C^2$. Similarly, for $\delta>0$ small enough no line in $\Lambda^2$ intersect $C^1$. This ensures that the union ${\cal R}^1 \cup {\cal R}^2$ with the point $(0,0)$ relabeled $x_0$ and the point $(0,1)$ relabeled $u$ is a realization of the join of $(\chi_1,u_1)$ and $(\chi_2,u_2)$. This proves that $\chi$ is realizable chirotope.
  
  It remains to check that $u$ is an extreme element in $\chi$,
   to ensure that $(\chi,u)$ is a rooted chirotope.
   But this follows immediately from the fact that $(0,1)$ is above all points of ${\cal R}^1 \cup {\cal R}^2$. This completes the proof of Proposition~\ref{prop:join_realizable}.
\end{proof}

\subsection{Twists and meets}

We now want to define the meet operator. Intuitively, it is similar to a join operator where the root is moved from ``the top'' to ``the bottom''. To avoid redefining it from scratch (and reproving its properties), we introduce another operator, called twist, that formalizes this displacement of the root.

\begin{definition}\label{def:twist}
  Let $(\chi, u)$ be a rooted chirotope with ground set $X \cup \{u\}$.
  We define its twist $\neg (\chi, u)$ as the rooted chirotope $(\wchi,v)$,
  where $\wchi$ has ground set $X \cup \{v\}$ and 
\[
\wchi (x,y,z) =\begin{cases}
  \chi(x,y,z) &\text{ if }x,y,z \in X;\\
  -\chi(x,y,u) & \text{ if }z=v.
\end{cases}
\]
\end{definition}

\begin{proposition}\label{prop:twist_realizable}
  In the above definition, $\wchi$ is indeed a chirotope. Moreover, if $\chi$ is realizable,
  then $\wchi$ is realizable as well.
\end{proposition}
\begin{proof}
  We first prove the statement for a rooted {\em realizable} chirotope $(\chi,u)$ on $X \cup \{u\}$. By 
\cref{lem:extreme_in_unbounded_cell},
there exists a realization $\mathcal P \cup \{\mathfrak p_u\}$
of $\chi$ for which $\mathfrak p_u$ is in an unbounded cell $C$
of the line arrangement ${\mathcal L}_{\mathcal P}$ formed by the points in $\mathcal P$.
Then there exists another cell of ${\mathcal L}_{\mathcal P}$, which we denote $\overline{C}$,
such that any point $\mathfrak p_v$ in $\overline{C}$ satisfies 
$\chi(\mathfrak q,\mathfrak r,\mathfrak p_v) 
= -\chi(\mathfrak q,\mathfrak r,\mathfrak p_u)$ for $\mathfrak q,\mathfrak r$ in $\mathcal P$.
Indeed, taking any half-line not parallel to any line of ${\mathcal L}_{\mathcal P}$ and whose unbounded extremity lies in $C$,
reversing its direction yields a half-line  whose unbounded extremity lies in $\overline{C}$.
Then, for any $\mathfrak p_v$ in $\overline{C}$,
 the point set $\mathcal P \cup \{\mathfrak p_v\}$ 
 is a realization of the chirotope $\wchi$ on $X \cup \{v\}$.

 It remains to prove the first statement for a chirotope $\chi$ on $X \cup \{u\}$ that is not realizable. We proceed as in the proof of \cref{cor:join}. Recall that $\wchi$ is a chirotope if and only if for every five-element subset $Y \subseteq X \cup \{u\}$ the restriction ${\wchi}_{|Y}$ is a chirotope. If $u \notin Y$, then ${\wchi}_{|Y} = \chi_{|Y}$ and it is clearly a chirotope. If $u \in Y$, then ${\wchi}_{|Y}$ coincides with the twist $\widetilde{\chi_{|Y}}$ of the restricted rooted chirotope $\chi_{|Y}$. (In other words, the operator $\widetilde{\cdot}$ commutes with the restriction to a subset $Y$ containing the root $u$.) The small size of $\chi_{|Y}$ ensures that it is realizable, and therefore so is its twist by the previous argument. In particular, $\widetilde{\chi_{|Y}} = {\wchi}_{|Y}$ is a chirotope. Altogether, every restriction ${\wchi}_{|Y}$ to a five-element subset is a chirotope and consequently so is $\wchi$. 
\end{proof}

Let us stress that the twist of a chirotope on $X \cup \{u\}$ is a chirotope on $X \cup \{v\}$, but {\bf not} a chirotope on $X \cup \{u,v\}$. This is a design choice: it avoids choosing how the line $uv$ partitions $X$. We can now define the meet operator.

\begin{definition}
Let $(\chi_1, u_1)$ and  $(\chi_2, u_2)$ be two rooted chirotopes. 
We define the \emph{meet} of these two chirotopes by
\[
(\chi_1, u_1)\wedge (\chi_2, u_2) = \neg \big(\neg (\chi_1, u_1) \vee \neg (\chi_2, u_2)\big).
\]
\end{definition}

Let us spell this out. Starting with two chirotopes on $X_1 \cup \{u_1\}$ and $X_2 \cup \{u_2\}$, we first compute their twists, which are chirotopes on $X_1 \cup \{v_1\}$ and $X_2 \cup \{v_2\}$, we then apply a join operator, obtaining a chirotope on $X_3 \cup \{v_3\}$, and we twist again, finally obtaining a chirotope on $X_3 \cup \{u_3\}$. Alternatively, the meet of two chirotopes can be defined as the join, except that \eqref{eq:def_chi3} is replaced by:

\begin{equation}
 \label{eq:def_chi3_bis}
 \chi_3(x,y,z) = 
  \begin{cases}
   \chi_1(x,y,z) \text{ \quad if } x,y,z \in X_1, \\
   \chi_2(x,y,z) \text{ \quad if } x,y,z \in X_2, \\
   -\chi_1(x,y,u_1) \text{ \quad if } x,y \in X_1 \text{ and } z \in X_2\setminus\{{u_2}^+\}, \\
   -\chi_2(u_2,y,z) \text{ \quad if } x \in X_1\setminus\{{u_1}^-\} \text{ and } y,z \in X_2;
   \end{cases}
 \end{equation}

\noindent
We now obtain the following as a direct consequence of \Cref{prop:join_realizable,prop:twist_realizable}.

\begin{proposition}
\label{prop:meet_realizable}
The meet $(\chi_3,u_3)$ of two rooted chirotopes $(\chi_1, u_1)$ and  $(\chi_2, u_2)$
is a rooted chirotope.
Moreover, if $(\chi_1, u_1)$ and  $(\chi_2, u_2)$ are realizable, 
then $(\chi_3,u_3)$ is realizable as well.
\end{proposition}

\begin{remark}\label{rmk:join=convexSum}
{Let $C$ be a chain in the sense of Rutschmann and Wettstein~\cite{rutschmann2023chains}, i.e.~a sequence of points
$(\mathfrak p_1,\dots,\mathfrak p_n)$ sorted in increasing $x$-coordinates 
such that edges of the form $\mathfrak p_i \mathfrak p_{i+1}$ belong to all triangulations of the corresponding point set.
Then we can take an additional point $\mathfrak p_u$ with a large $y$-coordinate
so that the triangle $\mathfrak p_i \mathfrak p_j \mathfrak p_u$ is in counterclockwise 
order for all $i,j$ with $i<j$.
We let $(\chi_C,u_C)$ be the corresponding rooted chirotope, where $u_C$
is the element indexing the point $\mathfrak p_u$.
This allows to see chains as particular cases of rooted chirotopes.}

{
Then it can be checked that if $C_1$ and $C_2$ are chains,
the join $(\chi_{C_1},u_{C_1}) \vee (\chi_{C_2},u_{C_2})$ is  the rooted
chirotope associated to the chain $C_1 \vee C_2$, that is the convex sum of $C_1$ and 
$C_2$ in the sense of~\cite{rutschmann2023chains}.
Indeed, it is straightforward to check that all triples are oriented in the same way 
in both chirotopes.
Similarly, the meet $(\chi_{C_1},u_{C_1}) \wedge (\chi_{C_2},u_{C_2})$
is the rooted chirotope associated to the concave sum $C_1 \wedge C_2$
of $C_1$ and $C_2$. In this sense, our join and meet operations are generalizations
to all chirotopes of the convex and concave sums of Rutschmann and Wettstein~\cite{rutschmann2023chains}.}
\end{remark}

\section{Counting triangulations of joins and meets}
\label{sec:triangulations_join}

We are interested in computing the number of triangulations of joins and meets of smaller chirotopes
(and of repeated iterations of joins and meets of smaller chirotopes).
To this end, counting triangulations recursively is not enough, we need to introduce a notion
of {\em weak triangulations}, and a bivariate polynomial, enumerating weak triangulations
with respect to the degree of some vertices.
In the following if $\Delta$ is a set of edges on a point set $X$ (often a triangulation),
 and $x$ an element of $X$,
we denote by $\deg_\Delta(x)$ the degree of $x$ in $\Delta$.

\subsection{Weak triangulations}

  Let $(\chi,u)$ be a rooted chirotope.
In the following, similarly to \cref{def:twist}, we consider an additional element $v$, 
and extend $\chi$ to triples in $X \cup \{v\}$
by setting $\chi(x,y,v)=-\chi(x,y,u)$ for $x,y$ in $X$. 
Note that $\chi$ defines chirotope structures on $X \cup \{u\}$ and on $X \cup \{v\}$,
but not on $X \cup \{u,v\}$, since $\chi(x,u,v)$ is undefined (for $x$ in $X$).

Nevertheless for a set $H$ of segments with endpoints in $X \cup \{u,v\}$ {\bf not containing} the segment $uv$,
one can decide whether it contains any crossing pair or not.
Indeed for $x,y$ in $X$, the equality $\chi(x,y,v)=-\chi(x,y,u)$
 implies that the segments $xu$ and $yv$ cannot be crossing.
Hence the potentially crossing pairs of segments contain
 at most one of the points $u$ or $v$,
and deciding whether they are crossing
 or not does not involve the undefined values $\chi(x,u,v)$.
Thus the following notion is well-defined.
\begin{definition}
  A weak triangulation of $(\chi,u)$ is a maximal non-crossing collection of segments with endpoints in $X \cup \{u,v\}$,
  {\bf not containing} the segment $uv$.
\end{definition}
We claim that any triangulation $\Delta$ of $\chi$ can be extended
in a unique way to a weak triangulation of $(\chi,u)$.
Indeed, it suffices to add to $\Delta$ all edges $vx$ with $x$ in $X$
that do not cross any edge in $\Delta$.
However, some weak triangulations are not extensions of triangulations of $\chi$.
This explains the terminology ``weak triangulations'': every triangulation is a weak triangulation, but the converse is not true.
See \cref{fig:weak_triangulations} for examples.
\begin{figure}[ht]
\[\includegraphics{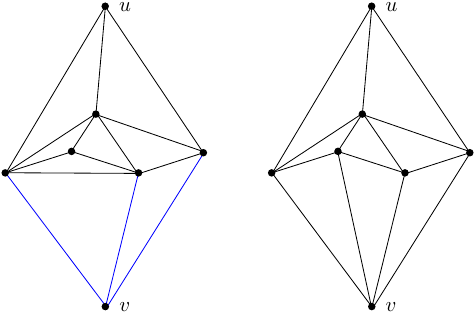}\]
\caption{Two weak triangulations of the same rooted chirotope $(\chi,u)$. The one on the left is an extension of a triangulation of $\chi$, but not the one on the right.}
\label{fig:weak_triangulations}
\end{figure}

The weak triangulations which are extensions of triangulations are easily identified 
through the following lemma.
\begin{lemma}\label{lem:identifying_true_triangulations}
Let $(\chi,u)$ be a rooted chirotope and let $m$ be the minimal degree of $v$ in a weak
triangulation of $(\chi,u)$. Then a weak triangulation $\Delta$ of $(\chi,u)$ is
the extension of a triangulation of $\chi$
if and only if $v$ has degree $m$ in $\Delta$.
\end{lemma}
\begin{proof}
By Euler formula, all weak triangulations $\Delta$ of $(\chi,u)$ have the same number of edges, say $M$.
Taking a weak triangulation $\Delta$ and removing its edges incident to $v$ gives a (not necessarily maximal) set of non crossing edges of $(\chi,u)$, which we denote $\Delta^o$.
The number of edges in $\Delta^o$ is $M-\deg_\Delta(v)$. The edge-set $\Delta^o$
is a triangulation if and only if its number of edges is maximal i.e.~if and only if $\deg_\Delta(v)=m$.
\end{proof}

\subsection{Polynomials counting weak triangulations}

For a rooted chirotope $(\chi,u)$, we denote
 $\mathfrak{W}(\chi,u)$ the set of weak triangulations of $(\chi,u)$.
 Moreover, we associate with $(\chi,u)$ a bivariate polynomial $P_{(\chi,u)}(u,v)$ defined by 
\[
P_{(\chi,u)}(u,v) = \sum_{\Delta \in \mathfrak{W}(\chi,u)} u^{\deg_\Delta(u)} v^{\deg_\Delta(v)},\]
where we recall that $\deg_\Delta(u)$ and $\deg_\Delta(v)$ denote respectively the number of edges incident to $u$ and $v$ in $\Delta$. 
Here, and in what follows, we use the same symbol for the element of the chirotope $u$ or $v$, 
and the associated variable in the bivariate polynomial.

Our goal in this section is to prove the following recursive formula for these bivariate polynomials. As usual, for any bivariate polynomial $P(u,v)$, we denote by $[u^a]P(u,v)$ the coefficient of $u^a$ in $P(u,v)$, which is a polynomial in $v$.

\begin{proposition}
  \label{prop:triangulation_join}
Let ${(\chi_1,u_1)}$ and  ${(\chi_2,u_2)}$ be two rooted chirotopes on at least three elements each. It holds that 
\[
 P_{{(\chi_1,u_1)} \vee {(\chi_2,u_2)}}(u,v) = \sum_{d_1 \geq 2} \sum_{d_2 \geq 2} [u^{d_1}] P_{{(\chi_1,u_1)}}(u,v) \cdot [u^{d_2}] P_{{(\chi_2,u_2)}}(u,v) \cdot N_{d_1,d_2}(u) v^{-1},
 \]
 where
$ \displaystyle N_{d_1,d_2}(u) \eqdef u^{d_1+d_2-1} + \sum_{i_1=1}^{d_1-1} \sum_{i_2=1}^{d_2-1} \binom{d_1-i_1+d_2-i_2-2}{d_1-i_1-1} u^{i_1+i_2}$.
\end{proposition}
Since the twist operation on rooted chirotopes simply switches the role of $u$ and $v$,
we have
\[P_{\neg (\chi,u)}(u,v) = P_{(\chi,u)}(v,u). \]
Hence Proposition~\ref{prop:triangulation_join} also gives a way to compute $P$ polynomials along meet operations. It suffices to switch the roles of $u$ and $v$ in Proposition~
\ref{prop:triangulation_join}.\medskip

For a rooted chirotope $(\chi,u)$, we also consider the generating polynomial $Q_{(\chi,u)}(u)$ of usual triangulations of $\chi$ with respect to the degree of $u$, 
i.e.~$Q_{(\chi,u)}(u)$ is defined by
\[Q_{(\chi,u)}(u)= \sum_{\Delta \in \mathfrak{T}(\chi)} u^{\deg_\Delta(u)},\]
where $\mathfrak{T}(\chi)$ is the set of triangulations of $\chi$.
The polynomial $Q_{(\chi,u)}(u)$ contains less information than $P_{(\chi,u)}(u,v)$.
Indeed, from Lemma~\ref{lem:identifying_true_triangulations}, one has
\[Q_{(\chi,u)}(u)=[v^m] P_{(\chi,u)}(u,v),\]
where $m$ is the minimal exponent of $v$ appearing in $P_{(\chi,u)}$. From \cref{prop:triangulation_join}, this yields a similar recursive formula for the polynomials $Q_{(\chi,u)}(u)$. 

\begin{corollary}
  \label{corol:Qjoin}
Let ${(\chi_1,u_1)}$ and  ${(\chi_2,u_2)}$ be two rooted chirotopes on at least three elements each. It holds that 
\[
 Q_{{(\chi_1,u_1)} \vee {(\chi_2,u_2)}}(u) = \sum_{d_1 \geq 2} \sum_{d_2 \geq 2} [u^{d_1}] Q_{{(\chi_1,u_1)}}(u) \cdot [u^{d_2}] Q_{{(\chi_2,u_2)}}(u) \cdot N_{d_1,d_2}(u),
\]
where, as in Proposition~\ref{prop:triangulation_join},
\[N_{d_1,d_2}(u) = u^{d_1+d_2-1} + \sum_{i_1=1}^{d_1-1} \sum_{i_2=1}^{d_2-1} \binom{d_1-i_1+d_2-i_2-2}{d_1-i_1-1} u^{i_1+i_2}.
\]
\end{corollary}

This corollary is useful when building larger chirotopes from smaller ones using only the join operation (as in \cref{sec:DC}); however, when combining joins and meets (as in \cref{sec:numeric}), the full generality of \cref{prop:triangulation_join} is necessary.

\begin{remark}\label{rmk:WeakTforChains=Upper*Lower}
{Let $C$ be a chain in the sense of Rutschmann and Wettstein~\cite{rutschmann2023chains},
and $(\chi_C,u_C)$ be the associated rooted chirotope, as defined in \cref{rmk:join=convexSum}. Then one can check that the edges $\mathfrak p_i \mathfrak p_{i+1}$ which by definition of chains belong to all triangulations of the chain, also belong to all weak triangulations of $(\chi_C,u_C)$.
It follows that a weak triangulation of $(\chi_C,u_C)$ can be decomposed into an upper part (which is a maximal set of non-crossing edges of the point set $\{\mathfrak p_1,\dots,\mathfrak p_n,\mathfrak p_u\}$ which are all above the chain)
and a lower part (which is a maximal set of non-crossing edges of the point set $\{\mathfrak p_1,\dots,\mathfrak p_n,\mathfrak p_v\}$ which are all below the chain).
Erasing $\mathfrak p_u$ and all incident edges from the upper part yields what 
Rutschmann and Wettstein call a {\em partial upper triangulation} of $C$,
while erasing $\mathfrak p_v$ and all incident edges gives a {\em partial lower triangulation} of $C$.
Moreover the degree of $u$ (resp.~$v$) corresponds to the number of visible edges in the partial upper (resp.~lower) triangulation of $C$, plus one.
Combining all these observations, we get that, for a chain with $n$ edges, the bivariate polynomial of the associated rooted chirotope $(\chi_C,u_C)$ (which has $n+2$ elements) writes as
\begin{equation}\label{eq:ForChians_PFactorizes}
P_{(\chi,u)}(u,v) = u^{n+1} T_C(u^{-1}) \times v^{n+1} T_{\bar C}(v^{-1}),
\end{equation}
where $T_C$ is the upper triangulation polynomial of $C$ (see~\cite[Definition 21]{rutschmann2023chains}), and $\bar C$ is the flip of the chain $C$.
With this formula, one can see that \cref{prop:triangulation_join} generalizes
the formula computing the upper/lower triangulation polynomials of convex sums \cite[Lemmas 22 and 24]{rutschmann2023chains}.}

{Equation~\eqref{eq:ForChians_PFactorizes} shows that, for chains, our bivariate polynomials factorize as the product of a polynomial in $u$ and a polynomial in $v$. 
Hence it suffices to compute two univariate polynomials instead of a bivariate polynomial.
This is not the case in general,
which explains why computations are more efficient in the case of chains.}
\end{remark}

\subsection{A bijection to prove \cref{prop:triangulation_join}}

We follow the same proof strategy as in \cite[Section 7]{sidmapapiercombi}. We first describe a decomposition of the weak triangulations of a chirotope of the form $(\chi_1,u_1)\vee  (\chi_2,u_2)$ into several pieces, and then we show that these pieces are exactly what is needed to reconstruct the weak triangulations of $(\chi_1,u_1)\vee  (\chi_2,u_2)$. The recursive formula will then follow from this bijection.

 We now describe the bijection and its inverse.
The bijection is illustrated in \cref{fig:join_decomposition}, and the reader is invited to follow the construction on the figure. We use the notation of of \cref{sec:Join_Meet}: in particular the points $u_1$ and $u_2$ (resp.~$v_1$ and $v_2$, and $u_1^-$ and $u_2^+$) are merged into a single point $u$ (resp.~$v$ and $x_0$) in $\chi$.
\begin{figure}[t]
\includegraphics[width=\linewidth]{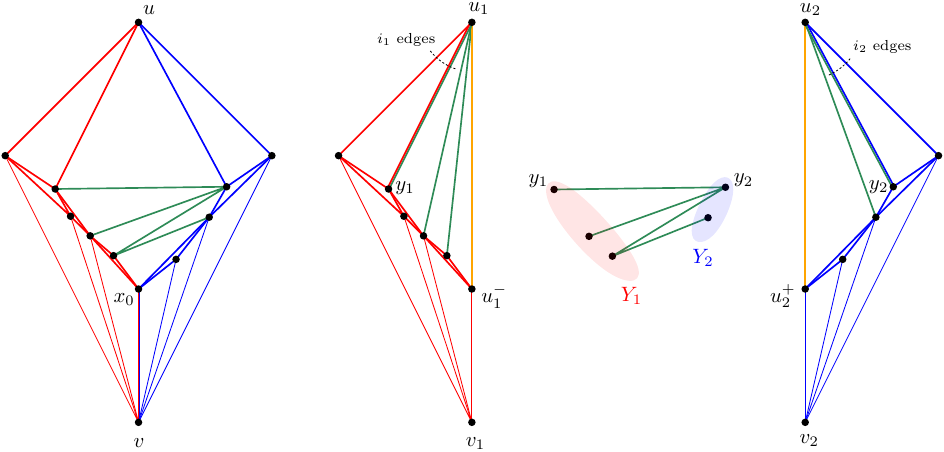}
\caption{The bijection used in the proof of \cref{prop:triangulation_join}. From left to right: $\Delta$, $\pi_1(\Delta)$, $C(\Delta)$ and $\pi_2(\Delta)$. \label{fig:join_decomposition}}
\end{figure} 

\paragraph*{Decomposition} 
Let $(\chi_1,u_1)$, $(\chi_2,u_2)$ be two rooted chirotopes of ground sets $X_1\cup \{u_1\}$, respectively $X_2\cup \{u_2\}$.
We consider the rooted chirotope $(\chi,u)=(\chi_1,u_1)\vee  (\chi_2,u_2)$, of ground set $X=X_1\cup X_2\cup \{u\}$,  and a weak triangulation $\Delta$ of $(\chi,u)$. We define $\pi_1(\Delta), \pi_2(\Delta), C(\Delta)$ as follows.
\begin{itemize}
	\item $C(\Delta)$ is the set of edges in $\Delta$ that have one endpoint in $X_1\setminus\{u_1^-\}$ and the other one in $X_2\setminus\{u_2^+\}$ (in green in the figure). Let us denote by $Y_1$ the endpoints of the edges of $C(\Delta)$ that are in $X_1$, and similarly for $Y_2$.
	\item $\pi_1(\Delta)$  contains three types of edges: 
	\begin{itemize}
	\item the set of edges from $\Delta$ with both endpoints in $X_1\cup \{u,v\}$ ,
	 where $u,v$ are relabeled $u_1,v_1$ in $\pi_1(\Delta)$ (in red in the figure);
	\item the edges of $C(\Delta)$ where every endpoint in $X_2\setminus\{u_2^+\}$ is replaced by $u_1$ (in green in the figure);
	\item the edge $u_1u_1^-$ (in orange in the figure).
		\end{itemize}
	Note that some edge can appear multiple times in the list above; in that case, we simply forget multiplicities.
	\item Similarly $\pi_2(\Delta)$  contains three types of edges: 
	\begin{itemize}
	\item the set of edges from $\Delta$ that have one endpoint in $X_2$ and the other one in $X_2\cup \{u,v\}$, where $u,v$ are relabeled $u_2,v_2$ in $\pi_2(\Delta)$ (in blue in the figure);
	\item the edges of $C(\Delta)$ where every endpoint in $X_1\setminus\{u_1^-\}$ is replaced by $u_2$ (in green in the figure);
	\item the edge $u_2u_2^+$ (in orange in the figure).
	\end{itemize}
\end{itemize}
As in \cite[Section 7]{sidmapapiercombi}, it can be proved that $\pi_1(\Delta), \pi_2(\Delta), C(\Delta)$ are sets of non-crossing edges (inherited from the non-crossing property of $\Delta$).
Furthermore, by a counting argument, it can be shown that  $\pi_1(\Delta)$ is  actually a weak triangulation of $(\chi_1,u_1)$, $\pi_2(\Delta)$ is a weak triangulation of $(\chi_2,u_2)$, and $C(\Delta)$ is a maximal set of non-crossing edges between $Y_1$ and $Y_2$.

\paragraph*{Converse construction}
To prove that this decomposition is a bijection, we describe the converse. Let $\Delta_1$ be a weak triangulation of $(\chi_1,u_1)$, $\Delta_2$ be a weak triangulation of $(\chi_2,u_2)$. Let us denote by $D_1$ (resp. $D_2$) the set of neighbors of $u_1$ in $\Delta_1$ (resp. of $u_2$ in $\Delta_2$). 

Note that $D_1$ can be equipped with a total order (the radial order with respect to $u_1$). Let $Y'_1\subseteq D_1$ be an interval for this order containing $u_1^-$, and set $Y_1=Y_1'\setminus \{u_1^-\}$. Furthermore, we denote by $y_1$ the other endpoint of $Y_1'$. Define $Y_2, Y_2', y_2$ symmetrically. Let $P$ be the set of such pairs $(Y_1,Y_2)$ such that either $Y_1=Y_2=\emptyset$, or none is empty.

 Starting from $\Delta_1$, $\Delta_2$, a pair $(Y_1,Y_2)\in P$  and a maximal set $C$ of non-crossing edges between $Y_1$ and $Y_2$, we build the following set $\Delta$ of edges in $(\chi, u)$:
\begin{enumerate}
	\item we add to $\Delta$ the edges of $\Delta_1$, except for the edges between $u_1$ and elements of $Y_1\setminus \{y_1\}$, and for the edge $u_1u_1^-$
	(replacing $u_1$, $u_1^-$, $v_1$ by $u$, $x_0$ and $v$ respectively);
	\item we add to $\Delta$ the edges of $\Delta_2$, except for the edges  between $u_2$ and elements of $Y_2\setminus \{y_2\}$, and the edge $u_2u_2^+$
	(replacing $u_2$, $u_2^+$, $v_2$ by $u$, $x_0$ and $v$ respectively);
	\item if $Y_1=Y_2=\emptyset$, we add the edge $ux_0$ (recall that $u_1^-$ and $u_2^+$ are identified as $x_0$ in $(\chi,u)$);\\
	otherwise we add the edges of $C$.
\end{enumerate}
We can check similarly to \cite[Section 7]{sidmapapiercombi} that $\Delta$ is indeed a weak triangulation of $(\chi,u)$, and that this construction is the converse of the previous decomposition.

\begin{remark}
\label{rem:multiple_edges}
The sets of edges added in steps (i) and (ii) of the above are disjoint except for the $vx_0$ which always appears in both.
Moreover, the edge $ux_0$ appears in $\Delta$ if and only if $Y_1=Y_2=\emptyset$.
\end{remark}

\paragraph*{Proof of \cref{prop:triangulation_join}}
 This bijection yields a way of counting the weak triangulations of a join.
Let ${(\chi_1,u_1)}$ and  ${(\chi_2,u_2)}$ be two rooted chirotopes. Recall that we want to prove that 
\[
 P_{{(\chi_1,u_1)} \vee {(\chi_2,u_2)}}(u,v) = \sum_{d_1 \geq 2} \sum_{d_2 \geq 2} [u^{d_1}] P_{{(\chi_1,u_1)}}(u,v) \cdot [u^{d_2}] P_{{(\chi_2,u_2)}}(u,v) \cdot N_{d_1,d_2}(u) v^{-1}, \text{ where}
\]
\[N_{d_1,d_2}(u) = u^{d_1+d_2-1} + \sum_{i_1=1}^{d_1-1} \sum_{i_2=1}^{d_2-1} \binom{d_1-i_1+d_2-i_2-2}{d_1-i_1-1} u^{i_1+i_2}.
\]

By definition, $[u^{d_1}] P_{{(\chi_1,u_1)}}(u,v)$ is a polynomial in $v$ counting, w.r.t. the number of neighbors of $v_1$, the weak triangulations of ${(\chi_1,u_1)}$ in which $u_1$ has $d_1$ neighbors (\emph{i.e.} such that $|D_1| =d_1$ with the notation above). Of course, we have a similar interpretation for $[u^{d_2}] P_{{(\chi_2,u_2)}}(u,v)$.

Once a pair $(Y_1,Y_2)$ of non-empty subsets of $D_1$ and $D_2$ has been chosen, the number of maximal sets of non-crossing edges between $Y_1$ and $Y_2$ is $\binom{|Y_1|+|Y_2|-2}{|Y_1|-1}$, as proved in~\cite[Lemma 7.2 and the sentence below it]{sidmapapiercombi}. In particular, it depends on $Y_1$ and $Y_2$ only through their sizes. 
The possible size of such a set $Y_1$ is at least $1$, and at most $d_1-1$ (as $u_1^- \notin Y_1$). Moreover, since $Y'_1=Y_1 \cup \{u_1^-\}$ must be an interval containing the extreme element $u_1^-$, there is exactly one possible set $Y_1$ of each size between $1$ and $d_1-1$.
We denote by $d_1-i_1$ the size of $Y_1$.  Then $i_1 \in \{1, \dots, d_1-1\}$ is the number of neighbors of $u_1$ that remain neighbors of $u$ when recombining the weak triangulations of ${(\chi_1,u_1)}$ and  ${(\chi_2,u_2)}$ into those of $(\chi,u)$. The same reasoning with $Y_2$ ensures that $u$ has degree $i_1+i_2$ in the recombined weak triangulation of $(\chi,u)$. This justifies the summation formula in the definition of $N_{d_1,d_2}(u)$.

For the pair $(Y_1,Y_2)$ of empty sets, \cref{rem:multiple_edges} ensures that $u$ has $d_1+d_2-1$ neighbors; indeed, the $-1$ accounts for the fact the edge $x_0u$ appears in $\Delta_1$ and $\Delta_2$ but should be counted only once in $\Delta$. This explains the first term of the definition of $N_{d_1,d_2}(u)$. 

As a consequence, for any weak triangulation $\Delta_1$ of $(\chi_1,u_1)$ is which $u_1$ has $d_1$ neighbors, and any weak triangulation $\Delta_2$ of $(\chi_2,u_2)$ is which $u_2$ has $d_2$ neighbors, $N_{d_1,d_2}(u)$ is a polynomial in the variable $u$ counting the number of weak triangulations $\Delta$ of $(\chi,u)$ that can be reconstructed from $\Delta_1$ and $\Delta_2$, where the degree of the variable $u$ indicates the number of neighbors of the element $u$ in $\Delta$. Note that this proves \cref{corol:Qjoin}. 

To complete the proof of \cref{prop:triangulation_join}, it remains to take into account the number of edges attached to vertex $v$ in $\Delta$. By construction, the edges attached to $v$ in $\Delta$ are exactly those attached to $v_1$ in $\Delta_1$ and those attached to $v_2$ in $\Delta_2$.
In the product $[u^{d_1}] P_{{(\chi_1,u_1)}}(u,v) \cdot [u^{d_2}] P_{{(\chi_2,u_2)}}(u,v)$ the degree of $v$ thus records the total number of such edges. However, the edge $vx_0$ appears in both $\Delta_1$ and $\Delta_2$ (see \cref{rem:multiple_edges}), explaining the factor $v^{-1}$ in the statement of \cref{prop:triangulation_join}. This concludes the proof of this proposition.

\section{Triangulations of the double circle}
\label{sec:DC}

We recall from the introduction that $\DC_k$ denotes the double circle with $k$ external (and thus $k$ internal) points.
The goal of this section is to establish a precise asymptotic estimate of its number of triangulations  $|\mathfrak{T}(\DC_k)|$.
\begin{theorem}\label{thm:triangulations_double_circle}
For large $k$, we have
\[ |\mathfrak{T}(\DC_k)|  \sim   
\frac{54} {7 \sqrt{\pi}} \, \frac{\sqrt{21}(5-\sqrt{21})}{(7-\sqrt{21})^2}\, 12^{k-2} \, k ^{-3/2}.\]
\end{theorem}

\subsection{The double circle through a series of joins}
Let $(\chi_1,u_1)$ be the rooted chirotope formed by four points, one being inside the triangle formed by the other three,
with an external point taken as the root.
We recursively define $(\chi_k,u_k)$ by $(\chi_{k+1},u_{k+1})=(\chi_k,u_k) \vee (\chi_1,u_1)$, for $k \ge 1$.
We note that $\chi_k$ can also be seen as the double circle $\DC_{k+2}$, with two consecutive internal points removed, the root $u_k$ being the external point
with no internal neighbors; see \cref{fig:chik}.
Let us denote by $Q_k=Q_k(u)$ the generating polynomial $Q_{(\chi_k,u_k)}$ of triangulations of $\chi_k$, where the exponent of $u$ denotes the degree of the marked element $u_k$, as introduced in \cref{sec:triangulations_join}.
Although $(\chi_k,u_k)$ is not exactly a double circle,
it turns out that the number of triangulations of the double circle
can be recovered from $Q_k$.
\begin{figure}
  \[ \includegraphics[scale=.9]{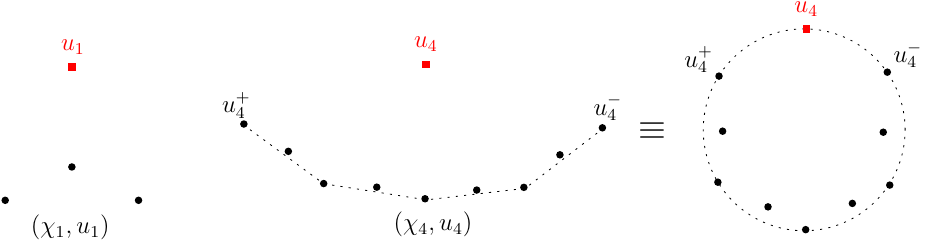}\]
  \caption{The chirotope $(\chi_k,u_k)$ is obtained from $(\chi_1,u_1)$ by iterating the join operation. It is the same chirotope as a double circle with two consecutive internal points removed.}
  \label{fig:chik}
\end{figure}
\begin{lemma}\label{lem:nbTriDoubleCircle}
  For $k \ge 2$, the number of triangulations of the double circle $\DC_k$ is given by
  \[ |\mathfrak{T}(\DC_k)| = |\mathfrak{T}(\chi_{k-1})|-|\mathfrak{T}(\chi_{k-1} \backslash \{u_{k-1}\})| = Q_{k-1}(1) - [u^2] Q_{k-1}(u). \]
\end{lemma}
\begin{proof}
The following argument is similar to those used by Hurtado and Noy
to count triangulations of almost-convex polygons (which include the double circle)~\cite{hurtado1997double-circle}.

Fix $k \ge 1$, and let us consider the set of triangulations of $\chi_k$, which are counted by 
$|\mathfrak{T}(\chi_{k})| = Q_k(1)$.
As usual, we call $u_k^+$ and $u_k^-$ the successor and predecessor of $u_k$ on the
convex hull of $\chi_k$ in counter-clockwise order.
Then $\mathfrak{T}(\chi_{k})$ can be partitioned into two subsets: 
triangulations containing the edge $u_k^+\, u_k^-$, or not.

The edge $u_k^+\, u_k^-$ is present if and only if $u_k$ has degree $2$ in the triangulation. 
Then removing $u_k$ and its incident edges 
(which are necessarily the edges $u_k\, u_k^-$ and  $u_k\, u_k^+$)
defines a bijection with triangulations of  $\chi_k \setminus \{u_k\}$.
Such triangulations are thus counted by 
$[u^2] Q_{k}(u)=| \mathfrak{T}(\chi_k \setminus \{u_k\}) |$.

Let us now consider a triangulation $T$ of $\chi_k$ which does not contain the edge $u_k^+\, u_k^-$. Then moving $u_k$ on the other side of the edge $u_k^+\, u_k^-$
(without changing orientations of other triples) and adding the edge $u_k^+\, u_k^-$
gives a triangulation of the resulting chirotope, which is  $\DC_{k+1}$.
This construction is a bijection, so that the triangulations $T$ of $\chi_k$ that do not contain the edge $u_k^+\, u_k^-$ are counted by $|\mathfrak{T}(\DC_{k+1})|$ .

Summing up, we have
\[|\mathfrak{T}(\chi_{k})| = Q_k(1) = [u^2] Q_{k}(u) + |\mathfrak{T}(\DC_{k+1})|
=| \mathfrak{T}(\chi_k \setminus \{u_k\}) | + |\mathfrak{T}(\DC_{k+1})|.\]
Changing $k$ to $k-1$, and reorganizing terms proves the lemma.
\end{proof}

\subsection{A functional equation}
Specializing the results of \cref{sec:triangulations_join},
we now obtain an induction formula for $Q_k$.
\begin{lemma}
  For $k \ge 1$ and $u \ne 1$, it holds that
  \[Q_{k+1}(u) = Q_k(u) \left( u^2+\frac{u^3}{(u-1)^2} \right)-Q_k(1) \, \frac{u^4-u^3+u^2}{(u-1)^2}-Q'_k(1) \, \frac{u^2}{u-1}.\]
  \label{lem:induction_Qk}
\end{lemma}

\begin{proof}
Recall that $Q_1(u)=Q_{(\chi_1,u_1)}(u)=u^3$.
 From \cref{corol:Qjoin}, for $k \ge 1$, it holds that
\begin{equation}\label{eq:rec_Qk}
  Q_{k+1}(u)=\sum_{d \ge 2} \big( [u^d] Q_k(u) \big) \cdot N_{d,3}(u),
\end{equation}
where \[N_{d,3}(u)=u^{d+2}+\sum_{i=1}^{d-1} \big( (d-i)u^{i+1} + u^{i+2} \big).\]
Elementary manipulations (which can be performed automatically, e.g., by Maple\footnote{We provide a maple worksheet both in \openfilelink{NoyauDoubleCerle.mw}{\underline{mw}} and \openfilelink{NoyauDoubleCerle.pdf}{\underline{pdf}} format with all the computations of this section. The files are embedded in this pdf (visible using, e.g., firefox) or can be found in the source of the arXiv version.}) 
give the following summation-free expression for $N_{d,3}(u)$ whenever $u \ne 1$:
\[N_{d,3}(u)= u^d \left( u^2+\frac{u^3}{(u-1)^2} \right) - \frac{u^4-u^3+u^2}{(u-1)^2} -d \frac{u^2}{u-1}.\]
Substituting into \eqref{eq:rec_Qk} and using that 
\[\sum_{d \ge 2} \big([u^d] Q_k(u)\big) \cdot u^d =Q_k(u),\quad \sum_{d \ge 2} \big([u^d] Q_k(u)\big)=Q_k(1),
\quad \sum_{d \ge 2} \big([u^d] Q_k(u)\big) \cdot d =Q'_k(1)\]
yields the announced formula. 
\end{proof}

We now introduce the generating series $F(z,u)=\sum_{k \ge 1} Q_k(u) z^k$.
Note that the degree of $u$ in every triangulation of $\chi_k$ is at most $2k+1$.
Hence, for a number $u \in \C$, we have 
\begin{equation*} \label{eq:expNbTriangulations}
 |Q_k(u)| \le |\mathfrak T(\chi_k)| \max(|u|,1)^{2k+1} = \mathcal O\big( (900 \max(|u|^2,1))^k \big),
\end{equation*}
where we used in the last inequality that any chirotope with $n$ elements has at most $30^n$ triangulations \cite{countingT}
(recall that $\chi_k$ has $2k+2$ elements).
Hence the above sum is well defined and analytic on the domain
\begin{equation*} \label{eq:domainD}
D=\big\{(z,u) \in \mathbb C^2, |z \max(|u|^2,1)| < \tfrac{1}{900} \big\},
\end{equation*}
which contains the point $(0,1)$.
We denote by $\partial_u F$ the partial derivative of $F$ with respect to $u$ on $D$.
The induction formula for $Q_k$ yields the following functional equation for $F$.
\begin{lemma}
  For any $(z,u)$ in $D$, we have
  \begin{equation}\label{eq:kernel}
    F(z,u) \, K(z,u) = u^3 (u-1)^2 z - (u^4-u^3+u^2) \, z \, F(z,1) - u^2 (u-1)\, \partial_u F(z,1),
  \end{equation}
  where $K(z,u)=(u-1)^2(1-zu^2)-zu^3$.
  \label{lem:functional_equation_F}
\end{lemma}
\begin{proof}
  For $u=1$, both sides reduce to $-zF(z,1)$ so that equality holds.
We now assume $u \ne 1$.
  We start from \cref{lem:induction_Qk}, multiply by $z^{k+1}$ and sum over $k \ge 1$.
  On the left-hand side, we have
  \[\sum_{k \ge 1} z^{k+1}  Q_{k+1}(u)= \left( \sum_{j \ge 1}  z^j Q_j(u) \right) -z Q_1(u) = F(z,u)-zu^3.\]
  On the right-hand side, we use the identities
  \[\sum_{k \ge 1} z^{k+1}  Q_{k}(u)=z F(z,u),\qquad
  \sum_{k \ge 1} z^{k+1}  Q_{k}(1)=z F(z,1),\qquad
   \sum_{k \ge 1} z^{k+1}  Q'_{k}(1)=z \partial_u F(z,1).\]
  Now, the identity given by \cref{lem:induction_Qk} yields
  \begin{equation}\label{eq:functionalF}
F(z,u)-zu^3 = z F(z,u)\, \left( u^2+\frac{u^3}{(u-1)^2} \right) 
  - z F(z,1)\, \frac{u^4-u^3+u^2}{(u-1)^2} - z\partial_u F(z,1) \frac{u^2}{u-1}.   
  \end{equation}
  The lemma follows by elementary manipulations.
\end{proof}

\subsection{Solving the functional equation}
Solving \cref{eq:kernel} can be done by the so-called kernel method,
the function $K$ being called the kernel.
The first step is to look for functions $u(z)$ such that $(z,u(z))$ is in $D$ for small $|z|$
and $K(z,u(z))=0$. We start by taking $z=x$ being a positive real number.
\begin{lemma}\label{lem:u1u2}
  For any real number $x$ in $(0,1/12)$, there exists unique real numbers $u_1=u_1(x)$ in $(1,2)$ and $u_2=u_2(x)$ in $(0,1)$
  such that $K(x,u_1(x))=K(x,u_2(x))=0$.
\end{lemma}
\begin{proof}
  Fix $x$ in $(0,1/12)$. Then $K(x,u)$ takes the following (limiting) values
  \begin{align*}
    \lim_{u \to-\infty} K(x,u)&=-\infty,\quad K(x,0)=1,\quad K(x,1)=-x<0,\\
    K(x,2)&=1-12x>0,\quad\lim_{u \to+\infty} K(x,u)=-\infty.
  \end{align*}
 Since $u\mapsto K(x,u)$ is a polynomial of degree 4,  
  it has exactly one root in each of these four intervals:
  $(-\infty,0)$, $(0,1)$, $(1,2)$ and $(2,+\infty)$. This proves the lemma.
\end{proof}

For small positive real numbers $x$, the quantities $F(x,1)$ and $\partial_u F(x,1)$ can be expressed in terms
of these implicit functions $u_1$ and $u_2$.
\begin{lemma}\label{lem:F_in_u1u2}
  For $x >0$ small enough, we have
  \begin{align}
    F(x,1)&=\frac{(u_1-1)(1-u_2)(u_1+u_2-1)}{u_1+u_2-u_1 u_2}; \label{eq:F_in_u1u2} \\ 
    \partial_u F(x,1)&=\frac{u_1u_2(u_1u_2-u_1-u_2+2) + {u_1}^2  + {u_2}^2 - 2u_1 - 2u_2 + 1}{u_1+u_2-u_1 u_2}. 
  \end{align}
  \label{lem:F_u1u2}
\end{lemma}
\begin{proof}
  For $x$ sufficiently small, since $u_1$ and $u_2$ are bounded,
  the pairs $(x,u_1(x))$ and $(x,u_2(x))$ lie in $D$.
  Therefore we can substitute them for $(z,u)$ in \eqref{eq:kernel}, and we get
  \[\begin{cases}
     0=u_1^3 (u_1-1)^2 x - (u_1^4-u_1^3+{u_1}^2) \, x \, F(x,1) - {u_1}^2 (u_1-1)\, \partial_u F(x,1);\\
     0=u_2^3 (u_2-1)^2 x - (u_2^4-u_2^3+{u_2}^2) \, x \, F(x,1) - {u_2}^2 (u_2-1)\, \partial_u F(x,1).
  \end{cases}\]
  Solving this linear system with a computer algebra software yields the announced expressions.
\end{proof}

\begin{remark}
  For $x$ in $(0,1/12)$ the kernel $K(x,u)$ has two other real roots $u_0(x)>2$ and $u_3(x)<0$.
  However, since $u_0$ and $u_3$ are unbounded, we have no guarantee that $(x,u_0(x))$ or $(x,u_3(x))$ are in $D$ for some $x$,
  so we cannot substitute those in the functional equations.
  This is standard when using the kernel method, and $u_1$ and $u_2$ are often called the \enquote{small roots}. 
  
  The four roots of the kernel are represented in~\cref{fig:rootsOfTheKernel}. 
\end{remark}

\begin{figure}[ht]
\begin{center}
\includegraphics[scale=0.5]{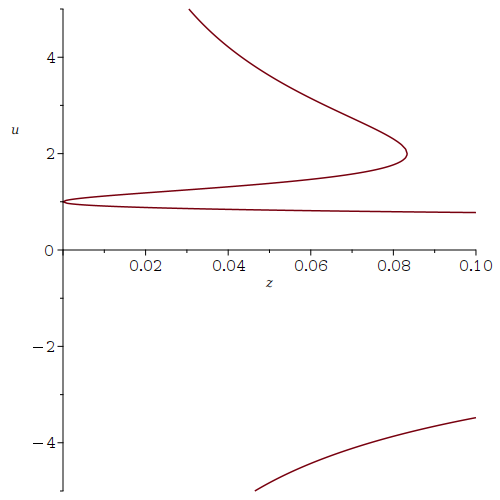}
\end{center}
\caption{A plot of the real solutions $u_0(x)$, $u_1(x)$, $u_2(x)$ and $u_3(x)$ (from top to bottom) of the kernel $K(x,u)=(u-1)^2(1-xu^2)-xu^3$, for $x$ real in $[0,0.1]$. The two roots $u_1$ and $u_2$ meet at the point $(x=0,u=1)$, while $u_0$ and $u_1$ meet at $(x=1/12,u=2)$. At these coalescence points, the two corresponding curves have squareroot singularities. \label{fig:rootsOfTheKernel}}
\end{figure}

\subsection{Analyticity and singular expansion of $F(z,1)$}

Our goal is to show that $F(z,1)$ is analytic on an appropriate domain
in order to apply analytic combinatorics techniques.
A comprehensive reference on analytic combinatorics with one variable
is the book of Flajolet and Sedgewick~\cite{Flajolet}.
 In our case, a difficulty arises because $u_1(z)$ and $u_2(z)$ have no analytic continuation at $z=0$: indeed, from the definition of $u_1(x)$ and $u_2(x)$ as the small roots of the kernel, for $i=1,2$ it can be seen that  $u_i(x) = 1 \pm \sqrt{x} (1+(o(1))$ for small positive $x$. 
Nevetheless, we have the following. 

\begin{lemma}\label{lem:u_i_analytic}
Fix $\theta \in (0,2\pi)$. Let us consider the domain 
\begin{align*}
 D(\theta) &= \{z : |z| < 4 \} \setminus \Big( \{r \exp(\mathrm{i} \theta) : r \geq 0\} \cup [1/12,+\infty) \Big).
\end{align*} 
For each $\theta$, it holds that $u_1$ and $u_2$ have analytic continuations on $D(\theta)$. Moreover, $u_2$ is analytic at $z=1/12$, and takes value $u_2(1/12)=(-3+\sqrt{21})/2$. On the other hand, for $z$ approaching $1/12$ in $D(\theta)$, we have 
\[
 u_1(z) = 2 - \sqrt{12/7} \cdot \sqrt{1-12z} + O(1-12z).
\]
\end{lemma}

\begin{proof}
Recall the definition of the kernel $K(z,u)=(u-1)^2(1-zu^2)-zu^3$. 
We introduce the characteristic system of the kernel equation: 
\[\begin{cases}
K(z,u)=0 \\
\partial_u K(z,u) = 0.
\end{cases}\]
The solutions $(z,u)$ of this system are the points where the implicit functions defined by the equation $K(z,u(z))=0$ have singularities,
because the analytic implicit function theorem does not apply.
Solving this system (e.g.~with Maple) give four solutions $(z,u) \in \C^2$, which are $(0,1)$, $(1/12,2)$ and $(-4,1/4\pm \mathrm{i}\sqrt{7}/4)$. 
Note that the domain $D(\theta)$ is chosen to be simply connected and not containing
the first (complex) coordinate of any of these four solutions.

If $(z_0,u_0)$ satisfies $K(z_0,u_0)=0$ but is not one of these pairs, then the analytic implicit function theorem \cite[Appendix B.5]{Flajolet} ensures that there exists an analytic root $u(z)$ of the kernel defined on a neighborhood of $z_0$ such that $u(z_0) =u_0$. 
Consequently, the functions $u_1$ and $u_2$ of \cref{lem:u1u2} can be analytically continued  along any curve included in $\mathbb{C} \setminus \{-4,0,1/12\}$. 
By the monodromy theorem, $u_1$ and $u_2$ have analytic continuations on the simply connected domain $D(\theta)$ for any $\theta \in (0,2\pi)$.

Moreover, when $x$ tends to $1/12$ from below, $u_2(x)$ converges to the unique root of $K(1/12,u)$ in $(0,1)$, which is $(-3+\sqrt{21})/2$. As $(1/12,(-3+\sqrt{21})/2)$ is not one of the four solutions of the characteristic system, $u_2$ can also be analytically continued on a neighborhood of $z=1/12$.  
On the other hand, when $x\to 1/12$, $u_1(x)$ converges to $2$; as $(1/12,2)$ is a solution of the characteristic system, the analytic implicit function theorem does not apply here. 
However, the singular implicit function theorem~\cite[Lemma VII.3 p.469]{Flajolet} ensures that, on a neighborhood of $z=1/12$ {\em without the half-line} $(1/12,+
\infty)$, one has 
\[
 u_1(z) = 2 \pm \gamma \cdot \sqrt{1-12z} + O(1-12z), 
\]
where
\[\gamma = \sqrt{\frac{2/12 \cdot \partial_z K(1/12,2)}{\partial_{uu} K(1/12,2)}} = \sqrt{\frac{2/12 \cdot (-12)}{(-7/6)}} = \sqrt{\frac{12}{7}}.\] 
Finally, to determine the sign, it is enough to notice that, when $z$ is real and $z < 1/12$, we have $u_1(z) <2$, finishing the proof of the lemma. 
\end{proof}

\begin{proposition}\label{prop:F}
There exists $R>1/12$ such that the functions $F(z,1)$ and $\partial_u F(z,1)$ admit analytic continuations on 
\[D_F= \{z : |z| < R \} \setminus [1/12,+\infty). \]
Moreover, their behaviors when $z$ approaches $1/12$ in $D_F$ are given by 
\begin{align}
  F(z,1) &= c_1 - c_2 \sqrt{1-12z} + O(1-12z), \label{eq:dvptFz1} \\
  \partial_u F(z,1) &= d_1 - d_2 \sqrt{1-12z} + O(1-12z), \label{eq:dvpt_deriveeFz1}
\end{align}
for 
\[c_1=\frac{-13+3\sqrt{21}}{7-\sqrt{21}},  \quad c_2 = \frac{12}{7} \frac{\sqrt{21}(5-\sqrt{21})}{(7-\sqrt{21})^2}, \quad d_1 =\frac{53-11\sqrt{21}}{7-\sqrt{21}} , \quad \text{and } d_2=4\,c_2.\]
\end{proposition}

\begin{proof}
The general idea of the proof is to use the \eqref{eq:F_in_u1u2} of $F$
in terms of $u_1$ and $u_2$ together with the local behavior of $u_1$ and $u_2$ around $z=1/12$ (given in Lemma~\ref{lem:u_i_analytic}) to find the behavior of $F$. Nevertheless, there are some subtleties to justify that $F$ is analytic on an appropriate domain, due in particular to the fact that we needed to remove a half-line in the definition of $D(\theta)$ to construct $u_1$ and $u_2$.

Until further notice, we use the analytic continuations of $u_1$ and $u_2$ given by \cref{lem:u_i_analytic} for $\theta=\pi$. These continuations are defined in $D(\pi)$, which contains the open interval $(0,1/12)$. 

We first show that the radius of convergence of $F(z,1)$ is at least $1/12$, and that the expression of $F(z,1)$ given in~\cref{eq:F_in_u1u2} is valid on a complex neighborhood of $(0,1/12)$. 

We start by the following observations: 
\begin{itemize}
 \item the exponential bound on the number of triangulations (see~\cref{eq:expNbTriangulations}) implies that $F(z,1)$ has positive radius of convergence; 
 \item the expression of $F(z,1)$ in~\cref{eq:F_in_u1u2} is also valid for $z$ a complex number of small enough modulus and not a negative real number (indeed, the proof of \cref{lem:F_in_u1u2} also works in this domain). 
\end{itemize}

Next, for $x \in (0,1/12)$, we have $0<u_2(x)<1<u_1(x)<2$, which ensures that $u_1+u_2 - u_1u_2 \neq 0$ on $(0,1/12)$. By continuity, this remains valid on a complex neighborhood of $(0,1/12)$. Therefore, $F(z,1)$ can be analytically continued on this complex neighborhood using~\cref{eq:F_in_u1u2}. In particular, $F(z,1)$ has no singularity on the open interval $(0,1/12)$. As $F(z,1)$ has nonnegative coefficients, it follows from Pringsheim's theorem that its radius of convergence is at least $1/12$. 

As a result, $F(z,1)$ is analytic on the open disk of radius $1/12$, and~\cref{eq:F_in_u1u2} is valid in this domain except for $z$ a nonpositive real number (recall that $u_1$ and $u_2$ are defined on $D(\pi)$, so in particular, they are not defined when $z$ is a nonpositive real number). 

Recall that $u_1(1/12)=2$ and $0<u_2(1/12)<1$ (see \cref{lem:u_i_analytic}). Hence, $u_1+u_2 - u_1u_2 \neq 0$ also holds on a slit neighborhood of $1/12$, \emph{i.e.} on a neighborhood of $1/12$ without the half-line $[1/12,+\infty)$ (recall that $u_1$ is not defined on this half-line). Consequently, $F(z,1)$ can be analytically continued using \cref{eq:F_in_u1u2} on this slit neighborhood of $1/12$. On this domain, the asymptotic expansion given in~\cref{eq:dvptFz1} follows by elementary computations from the expansions of $u_1(z)$ and $u_2(z)$ in $z=1/12$  given in~\cref{lem:u_i_analytic}. 
(The computations are cumbersome and were performed with Maple).

It remains to prove that $F(z,1)$ is analytic on $D_F$. 
The right-hand side of~\cref{eq:F_in_u1u2} is analytic on $D(\pi)$, except on the (potentially empty) set $\{z : u_1(z)+u_2(z) - u_1(z)u_2(z) = 0\}$. In particular, it coincides with $F(z,1)$ on 
\[
 \{z : |z| < 1/12\} \setminus \Big( (-\infty,0] \cup \{z : u_1(z)+u_2(z) - u_1(z)u_2(z) = 0\} \Big).
\]

Assume for the sake of contradiction that there exists $z_0$ of modulus at most $1/12$ such that $u_1(z_0)+u_2(z_0) - u_1(z_0)u_2(z_0) = 0$. Then, there exists a sequence $(z_n)_{n\geq1}$ of complex numbers of modulus less than $1/12$ such that $z_n \to z_0$ and $|F(z_n,1)| \to +\infty$. 
But since $F(z,1)$ has nonnegative coefficients, we have $|F(z_n,1)| \leq F(|z_n|,1)$, which leads to a contradiction since $F(x,1)$ is bounded on $[0,1/12]$
(here, we can take the closed interval $[0,1/12]$ since $F(x,1)$ converges when $x$ tends to $1/12$, see \cref{eq:dvptFz1}). Therefore, $u_1(z)+u_2(z) - u_1(z)u_2(z) \neq 0$ on the slit disk $\{z : |z| \leq 1/12\} \setminus (-\infty,0]$. This allows to continue $F(z,1)$ analytically on a neighborhood of each $z_0$ of modulus $1/12$, except $z_0=-1/12$ and $z_0=1/12$. 

Changing the value of $\theta$ (from $\pi$ to $\pi/2$ for instance), the same reasoning ensures that $F(z,1)$ can be analytically continued on a neighborhood of $-1/12$.
Summing up it  can be analytically continued on a neighborhood of any $z_0$ of modulus $1/12$, except $z_0=1/12$. Since we proved earlier that $F(z,1)$ has an analytic continuation on a slit neighborhood of $1/12$, this shows that $F(z,1)$ has an analytic continuation on a domain $\{z : |z| < R \} \setminus [1/12,+\infty)$ for some well-chosen $R>1/12$. 

This concludes the proof for $F(z,1)$. The proof of $\partial_u F(z,1)$ follows the exact same steps. 
\end{proof}

\subsection{Asymptotic number of triangulations of the double circle}

From~\cref{lem:nbTriDoubleCircle}, to compute the asymptotic number of triangulations of the double circle, we are interested in asymptotic expansions of $Q_{k-1}(1)$ and $[u^2]Q_{k-1}(u)$. Recall that $F(z,u) = \sum_{k\geq 1} Q_k(u) z^k$. 

The quantity $Q_{k-1}(1)$ is immediate to estimate. Applying the transfer theorem (see~\cite[Corollary VI.1 p.392]{Flajolet}) to the expression established by \cref{prop:F}, we get
\[
 Q_{k-1}(1) = [z^{k-1}] F(z,1) \sim \frac{c_2}{2 \sqrt{\pi}} \, 12^{k-1} \, k ^{-3/2}.
\]

For the quantity $[u^2]Q_{k-1}(u)$, keeping in mind that $F(z,0)=0$, we observe that taking the coefficient of $u^2$ in~\cref{eq:functionalF} gives 
\[
 [u^2]F(z,u) = -zF(z,1) + z\partial_u F(z,1).
\]
Applying again the transfer theorem, we obtain 
\begin{align*}
 [u^2]Q_{k-1}(u) &= -[z^{k-2}] F(z,1) + [z^{k-2}]\partial_u F(z,1) \\
 &\sim \frac{d_2-c_2}{2 \sqrt{\pi}} \, 12^{k-2} \, k ^{-3/2} = \frac{3c_2}{2 \sqrt{\pi}} \, 12^{k-2} \, k ^{-3/2}.
\end{align*}

Finally, from \cref{lem:nbTriDoubleCircle}, 
\begin{align*}
|\mathfrak{T}(\DC_k)|  &= Q_{k-1}(1) - [u^2] Q_{k-1}(u) \sim \frac{9c_2}{2 \sqrt{\pi}} \, 12^{k-2} \, k ^{-3/2}.
\end{align*}
This completes the proof of \cref{thm:triangulations_double_circle}.

\section{About numerical computations}
\label{sec:numeric}
In this section, we shortly report on a (failed) attempt to find chirotopes with more triangulations than the Koch chain. Let us first recall the construction of the Koch chain by Rutschmann and Wettstein~\cite{rutschmann2023chains}, presented in our language of rooted chirotopes with joins and meets.
We start with the rooted chirotope $\bar K_0$ with three points (one being the root).
Then define successively for odd $i$'s $\bar K_i = \bar K_{i-1} \vee \bar K_{i-1}$,
and for even $i$'s $\bar K_i=\bar K_{i-1} \wedge \bar K_{i-1}$. The Koch chains $K_i$ of Rutschmann and Wettstein then corresponds to $\bar K_i$ without its root element.

For large size chirotopes, the Koch chains are the known chirotopes with the largest number of triangulations. However, for small size, we can find chirotopes with more triangulations. For example, for $i=3$, the Koch chain $K_3$ has $n=9$ elements and $424$ triangulations, but there exists a chirotope $\chi^{\text{MAX}}_9$ with 9 elements and $729$ triangulations~\cite{AAK02}. Our idea was therefore the following: if, in the recursive construction of the Koch chain, we replace $\bar K_3$ by a chirotope with the same number of elements (here, 10, because $\bar K_3$ is $K_3$ with an additional root), but more triangulations, we may obtain chirotopes of larger size with more triangulations than the Koch chains of the corresponding size.\medskip

We therefore made the following experiment. Using the small chirotopes database of Aichholzer and co-authors~\cite{AAK02}, we considered all chirotopes with 10 points with all possible roots (which must be an extreme point). For each such rooted chirotope $(\chi,u)$,
 we constructed recursively a sequence $\bar K_i^{(\chi,u)}$ by taking $\bar K_3^{(\chi,u)}=(\chi,u)$, and using the same induction as for the Koch chain
 (i.e.~for odd $i$'s $\bar K_i = \bar K_{i-1} \vee \bar K_{i-1}$,
and for even $i$'s $\bar K_i=\bar K_{i-1} \wedge \bar K_{i-1}$).
Using the formulas given in Proposition~\ref{prop:triangulation_join}, we compute the bivariate polynomial associated to $\bar K_i^{(\chi,u)}$, up to $i=7$, and the number of weak triangulations\footnote{A small computational trick allows to compute this number without fully computing the last bivariate polynomial.} of $\bar K_8^{(\chi,u)}$. It turns out that the largest number of weak triangulations is obtained when $(\chi,u)$ is the Koch chain $\bar K_3$, suggesting that the Koch chain is optimal, at least with this kind of construction. Data related to these computations are available on request.
We did similar types of experiments with other constructions allowing to compute recursively the numbers of triangulations, but never manage to find chirotopes with more triangulations than the Koch chain.

\bibliographystyle{bibli_perso}
\newcommand{\etalchar}[1]{$^{#1}$}

\end{document}